\documentclass[twoside,10pt]{amsart}

\usepackage{graphicx,amsmath,amsfonts,amsthm,amssymb, enumerate}
\newtheorem{thm}{Theorem}[section]

\newtheorem{lem}[thm]{Lemma}
\newtheorem{prop}[thm]{Proposition}
	
\theoremstyle{definition}
\newtheorem{defn}[thm]{Definition}	
\theoremstyle{remark}

\newcommand{\im}{\operatorname{Im}}
\newcommand{\re}{\operatorname{Re}}

\newcommand{\cL}{\mathcal{L}}
\newcommand{\lspan}{\operatorname{span}}

\newcommand{\bth}{\boldsymbol{\theta}{}}
\newcommand{\ve}{\boldsymbol{e}}
\newcommand{\bom}{\boldsymbol{\omega}}
\newcommand{\bOm}{\boldsymbol{\Omega}}
\newcommand{\bara}{{a^*}}
\newcommand{\bzeta}{\zeta^*}
\newcommand{\barZ}{Z^*}
\newcommand{\barA}{{A^*}}
\newcommand{\barB}{B^*}
\newcommand{\barC}{C^*}
\newcommand{\barL}{L^*}
\newcommand{\barmu}{\mu^*}
\newcommand{\barnu}{\nu^*}
\newcommand{\baral}{{\alpha^*}}
\newcommand{\bardel}{{\delta^*}}
\newcommand{\barX}{{X^*}}

\newcommand{\rA}{\mathrm{A}}
\newcommand{\rB}{\mathrm{B}}
\newcommand{\rC}{\mathrm{C}}

\newcommand{\rE}{\mathrm{E}}
\newcommand{\rL}{\mathrm{L}}
\newcommand{\rP}{\mathrm{P}}
\newcommand{\BP}{\operatorname{BP}}
\newcommand{\CP}{\operatorname{CP}}
\newcommand{\AP}{\operatorname{AP}}
\renewcommand{\AE}{\operatorname{AE}}
\newcommand{\AL}{\operatorname{AL}}
\newcommand{\BL}{\operatorname{BL}}
\newcommand{\BE}{\operatorname{BE}}
\newcommand{\CE}{\operatorname{CE}}

\renewcommand{\d}{\mathrm{d}}
\newcommand{\e}{\mathrm{e}}
\renewcommand{\i}{\mathrm{i}}

\newcommand{\tC}{\tilde{C}}
\newcommand{\tX}{\tilde{X}}
\newcommand{\tDelta}{\tilde{\Delta}}
\newcommand{\tUpsilon}{\tilde{\Upsilon}}
\newcommand{\tnu}{\tilde{\nu}}
\newcommand{\tz}{\tilde{z}}

\newcommand{\hz}{\hat{z}}
\newcommand{\hDelta}{\hat{\Delta}}
\newcommand{\hUpsilon}{\hat{\Upsilon}}
\newcommand{\hnu}{\hat{\nu}}

\newcommand{\al}{\alpha}

\begin{document}

\title{Invariant classification of vacuum PP-waves}
\author{R. Milson, A. Coley, D. McNutt} 

\address{\ Dept. Mathematics and Statistics, Dalhousie U., 
Halifax, Nova Scotia B3H 4R2, Canada}

\email{rmilson@dal.ca, aac@mathstat.dal.ca, ddmcnutt@dal.ca}

\begin{abstract} 
  We solve the equivalence problem for vacuum PP-wave spacetimes by
  employing the Karlhede algorithm. Our main result is a suite of
  Cartan invariants that allows for the complete invariant
  classification of the vacuum pp-waves.  In particular, we derive the
  invariant characterization of the $G_2$ and $G_3$ sub-classes in
  terms of these invariants.  It is known \cite{Collins91} that the
  invariant classification of vacuum pp-waves requires \emph{at most}
  the fourth order covariant derivative of the curvature tensor, but
  no specific examples requiring the fourth order were known.  Using
  our comprehensive classification, we prove that the $q\leq 4$ bound
  is sharp and explicitly describe all such maximal order solutions.
\end{abstract} 

\maketitle

\section{Introduction}
In general relativity, identical spacetimes are often given in
different coordinate systems, thereby disguising the diffeomorphic
equivalence of the underlying metrics.  It is consequently of
fundamental importance to have an invariant procedure for deciding the
question of metric equivalence.  One approach to this problem is to
utilize scalar curvature invariants, obtained as full contractions of
the curvature tensor and its covariant derivatives \cite{ACSI}.
However, a particularly intriguing situation arises when we consider
pp-waves, space-times that admit a covariantly constant null vector
field \cite[Chapter 24]{ES}..  Some time ago it was observed that
all curvature invariants of a pp-wave spacetime vanish \cite{schmidt}.
Subsequently all space-times with the VSI property (vanishing scalar
invariants) and the more general CSI property (constant scalar
invariants) were classified \cite{VSI,BCSI}.  It is now known that
either a spacetime is uniquely determined by its scalar curvature
invariants, or is a degenerate Kundt spacetime \cite{ACSI,DKundt}; the
VSI and CSI solutions belong to this more general class.

To invariantly classify the degenerate Kundt spacetimes, and pp-waves
in particular, one must therefore use the Karlhede algorithm
\cite{karlhede} \cite[Chapter 9.2]{ES}, which is the Cartan
equivalence method \cite{Cartan} adapted to the case of 4-dimensional
Lorentzian manifolds.  The invariant classification proceeds by
reducing the 6-dimensional Lorentz frame freedom by normalizing the
curvature tensor $R$ and its covariant derivatives, $R^q$.  The
unnormalized components of $R^q$ are called \emph{Cartan
  invariants}. We define the \emph{IC (invariant classification)
  order} of a given metric to be the highest order $q$ required for
deciding the equivalence problem for that metric.  An upper bound on
the IC order is often referred to as the \emph{Karlhede bound}.

% The Karlhede algorithm proceeds by successively normalizing the
% curvature tensor and its covariant derivatives.  
% It is customary to
% set $t_{-1}= 0$ and $d_{-1} = 6$, where the latter is the dimension of
% the group of 4D Lorentz transformations.  
Set $t_{-1}=0$ and $d_{-1}=6$ (the dimension of the Lorentz group).
At each order $q\geq 0$, let $0\leq t_{q-1}\leq t_q$ denote the number
of functionally independent Cartan invariants and let $6\geq
d_{q-1}\geq d_q$ denote
% the degrees of frame freedom remaining after
% said normalizations; i.e; $d_q$ is \
the dimension of the joint isotropy group of the normalized $R,
R^1,\ldots, R^q$.  The algorithm terminates as soon as $t_{q-1} = t_q$
and $d_{q-1} = d_q$. A value of $d_q=0$ means that there exists an
invariant tetrad.  If $t_q<4$, then Killing vectors are present.  The
dimension of the isometry group is $4-t_q+d_q$. Henceforth, we will
refer to the sequence $(t_0,t_1,\ldots, t_q)$ as the \emph{invariant
  count}.

In this paper, we focus on a particularly simple class of VSI
spacetimes: the vacuum pp-waves, whose metric has the simple form
shown in equation \eqref{eq:vsinvacmetric} below.  The symmetry
classes for pp-waves were initially classified by Kundt and Ehlers
\cite{KE62} \cite[Table 24.2]{ES} for vacuum solutions, and
subsequently extended by Sippel and Goenner \cite{sipgo86} to the
general case.  The Karlhede bound for pp-waves was investigated in
\cite{Collins91} and \cite{ramvick96} where $q\leq 4$ was established;
however, it was not known whether this bound is sharp, or if it could
be lowered further. Despite the fact that these metrics have a very
simple form, depending on just one parametric function $f(\zeta,u)$
(see equation \eqref{eq:vsinvacmetric} below), the present paper is
the first to present a complete invariant classification for vacuum
pp-waves, and to establish the sharpness of the $q\leq 4$ bound.

\begin{figure}[ht]
  \centering
  \includegraphics[width=10cm]{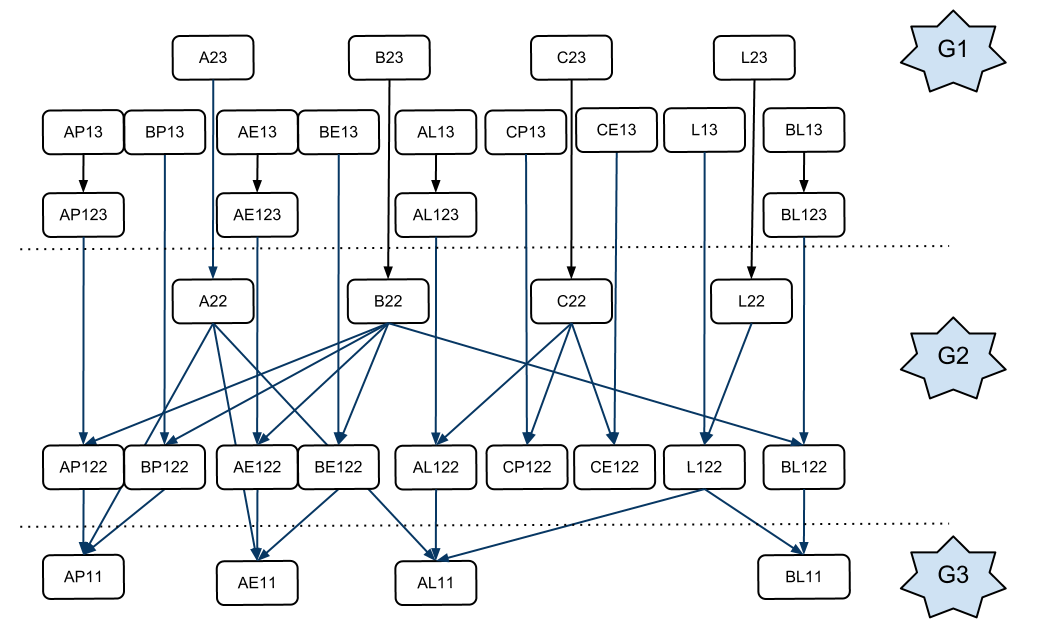}
  \caption{Specialization of  $G_1\to G_2\to G_3$ solutions in the
    $\alpha\neq 0$ class.}
  \label{fig:g23}
\end{figure}

All vacuum pp-waves have at least one Killing vector.  Kundt and
Ehlers identified 3 classes of $G_2$ solutions, 4 classes of $G_3$
solutions, a universal form for the $G_5$ solutions, and two types of
homogeneous $G_6$ solutions.  Below, we exhibit explicit Cartan
invariants that distinguish the various special sub-classes in an
invariant fashion.  

The $G_1,G_2, G_3$ solutions ($\alpha\neq 0$) and the $G_5, G_6$ solutions
($\alpha=0$) form two distinct solution branches; here $\alpha$ is a
fundamental 1st order invariant which will be defined precisely in
Section \ref{sect:ppwaves}.  The classification of the $\alpha\neq 0$
class is summarized in Figure \ref{fig:g23}.  The numbers in the
solution labels refer to the invariant count with the initial 0 and
any trailing $3$ omitted.  Thus, solution form $\AP_{123}$ refers to a
metric with an invariant count of $(0,1,2,3,3)$ while $\AP_{122}$
refers to a $G_2$ solution with an invariant count of $(0,1,2,2)$.
The $G_1$ solutions have three independent invariants and thus their
label indices end with a $3$.  For the same reason, the indices of the
$G_2$ solutions end with a 2 while the indices of the $G_3, G_5$
solutions end with a $1$.

From the point of view of invariant classification there are 4 classes
of generic $G_2$ solutions. We label these $\rA_{22}, \rB_{22},
\rC_{22}, \rL_{22}$ and summarize their invariant classification in
Table \ref{tab:022} (the Cartan invariants in the third column will be
defined in Section \ref{sect:g1}.)  Kundt-Ehlers described forms
$\rB_{22}$ and $\rL_{22}$.  Their third $G_2$ form is
\begin{equation}
  \label{eq:CA22}
 f(\zeta,u) = F(\zeta \e^{\i ku}),
\end{equation}
where $F$ is a holomorphic function and $k$ a real constant.  The $k$
parameter is not essential, and if $k\neq 0$ can be normalized to
$k\to 1$ by means of a coordinate transformation.  In terms of the
present terminology, the Kundt-Ehlers solutions of type
\eqref{eq:CA22} belong to class $\rC_{22}$ in the the case of $k=1$,
and and to class $\rA_{22}$ if $k=0$.

\begin{table}[ht]
  \centering
  \begin{tabular}{|@{\vbox to 10pt{} } c|c|l|l|}
    \hline
    $G_2$ & $f(\zeta,u)$ & Invariant condition \\ \hline
    $\rB_{22}$ & $F( u^{-ik}\zeta)u^{-2}$ &$B_2/B_1= k,\;  \Delta X_1
    = 2X_1^2,\; \hUpsilon=0$,  $AA^*\neq 1,B_1\neq 0$  \\
    $\rC_{22}$ &$F(\zeta \e^{\i u})$ & $B_1=0,\Delta X_2=0,\; \hUpsilon=0,\;
    \mu\neq 0, AA^*\neq 1$ \\
    $\rL_{22}$ & $ g\log \zeta$ & $A=1,\, Y=0$\\  
    $\rA_{22}$ & $F(\zeta)$ & $\mu= 0, \Delta \nu=0$\\
    \hline
  \end{tabular}
  \caption{Type $(0,2,2)$   $G_2$ solutions}
  \label{tab:022}
\end{table}

One benefit of the invariant classification is a clear description of
the mechanism of specialization of the $G_1\to G_2 \to G_3$ solutions.
In order to understand the $G_1\to G_2$ specialization one first has
to understand the invariant mechanism by which the solution forms in
Table \ref{tab:022} arise.  To that end, we show in Proposition
\ref{prop:genform} that all of vacuum pp-wave solutions of interest
can be reduced to the following form
\begin{equation}
  \label{eq:A**23}
  f(\zeta,u) = g_1 F(g_2 \zeta) + g_3 \zeta,
\end{equation}
where $F$ is a holomorphic functions and where $g_i=g_i(u),\; i=1,2,3$
are complex valued functions of one variable.  This general ansatz,
which we name $\rA^{**}_{23}$, bifurcates into a number of more
specialized forms, which are summarized in Table \ref{tab:g1} of the
Appendix.  Roughly speaking, there are 6 solution forms, which we
label by A,B,C, P,E,L and by numerical indices that describe the
invariant count.  An asterisk denotes a generic precursor of a more
specialized solution.  The labels P,E,L refer to, respectively,
solutions of power, exponential and logarithmic type.  Roughly
speaking, the Kundt-Ehlers $G_2$ solution forms are appropriate
specializations of the A,B,C and L solution forms.

The $G_1\to G_2$ specialization can be understood via the notion of a
``precursor solution''.  This is a $G_1$ solution that is mild
generalization of a corresponding $G_2$ solution.  For example the
precursor of the $\rB_{22}$ solution
\[ f(\zeta,u) = F(u^{-ik} \zeta) u^{-2} \]
is the $\rB_{23}$ solution 
\begin{equation}
  \label{eq:B23}
  f(\zeta,u) = F(u^{-ik} \zeta) u^{-2} + g \zeta,
\end{equation}
where $g=g(u)$ is an arbitrary complex valued function of one
variable.  Precursors of the other $G_2$ solutions have an analogous
form.  The invariant conditions that define the various precursor
classes are listed in Table \ref{tab:preg202} of the Appendix.  In
each case, the specialization to a $G_2$ involves the loss of the
$g\zeta$ term, or equivalently, the vanishing of a certain higher
order invariant.

As we show below, a vacuum pp-wave has no zeroth order invariants
\cite{ES}, and generically two independent first order invariants,
$\alpha,\alpha^*$.  In order to understand the $G_2\to G_3$
specialization it is necessary to understand the sub-class of
solutions for which $t_1=1$; i.e, metrics for which the invariants
$\alpha$ and $\alpha^*$ are functionally dependent. We refer to such
solutions as belonging to the (0,1) class and devote Section
\ref{sect:01} to their analysis.  Thus, the specialization to the
$G_3$ solutions follows the following path:
\[ (0,1,3) \to (0,1,2,2) \to (0,1,1) \] where the middle step consists
of type (0,1) $G_2$ solutions; summarized in Table \ref{tab:0122} of
the Appendix.

Another consequence of our analysis is a firm determination of the
Karlhede bound for vacuum pp-waves.  It turns that $q\leq 4$ is the
sharp bound.
\begin{thm}\label{thm:maxorder}
  There exist vacuum pp-wave spacetimes with an IC order $q=4$.  Every
  such metric belongs to one of the four classes exhibited in Table
  \ref{tab:0123}.
\end{thm}
\noindent
Note that metrics that require 4th order invariants for invariant
classification necessarily have a (0,1,2,3,3) as their invariant
count.

\begin{figure}[ht]
  \centering
\includegraphics[width=10cm]{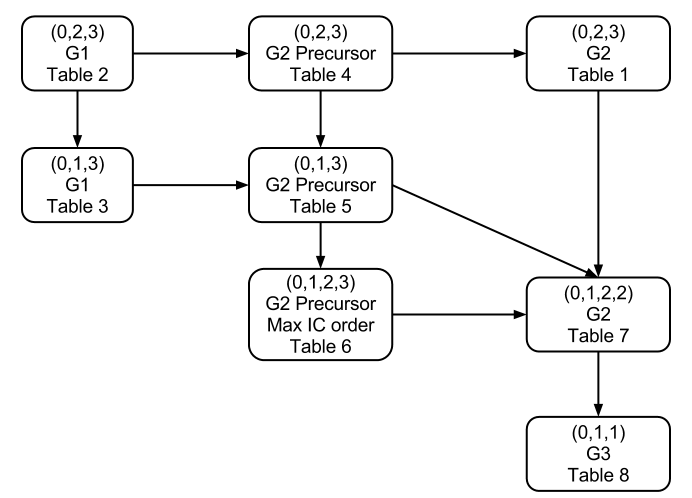}
\caption{The invariant classification of the $\alpha\neq 0$ class.}
  \label{fig:g123}
\end{figure}

The rest of the paper is organized as follows.  Section
\ref{sect:ppwaves} is an introductory description of the Karlhede
algorithm as it applies to the class of vacuum pp-wave metrics.  In
particular, this section describes the fundamental bifurcation into
the generic $\alpha\neq 0$ class and the specialized $\alpha=0$
subclass.  The invariant classification of the former consists of 8
sub-class types shown in Figure \ref{fig:g123}. Section \ref{sect:g1}
introduces the various Cartan invariants necessary for the generic
classification and derives the A,B,C,P,E,L solution forms in an
invariant manner. Section \ref{sect:01} deals with the type (0,1)
solutions in the $\alpha\neq 0$ class.  Section \ref{sect:g2precursor}
classifies the $G_2$-precursor solutions.  Section \ref{sect:0123}
derives and classifies the $G_1$ metrics having maximal IC order; the
proof of Theorem \ref{thm:maxorder} is given here.  Sections
\ref{sect:g1}, \ref{sect:01}, \ref{sect:g2precursor}, \ref{sect:0123},
when taken together, constitute the invariant classification of the
$G_1$ solutions; the specialization diagram for the various $G_1$
sub-classes is presented in Figure \ref{fig:g1} of the Appendix.
Sections \ref{sect:g2} and \ref{sect:g3} deal with the invariant
classification of the $G_2$ and $G_3$ solutions, respectively.
The $\alpha=0$ branch consists of $G_5$ and $G_6$ solutions. There is
a generic $G_5$ solution that specializes into two distinct classes of
homogeneous $G_6$ solutions, as per Figure \ref{fig:g56}.  This branch
of the classification is discussed in Section \ref{sect:g56} and
summarized in Table \ref{tab:g56}.

\begin{figure}[ht]
  \centering
  \includegraphics[width=10cm]{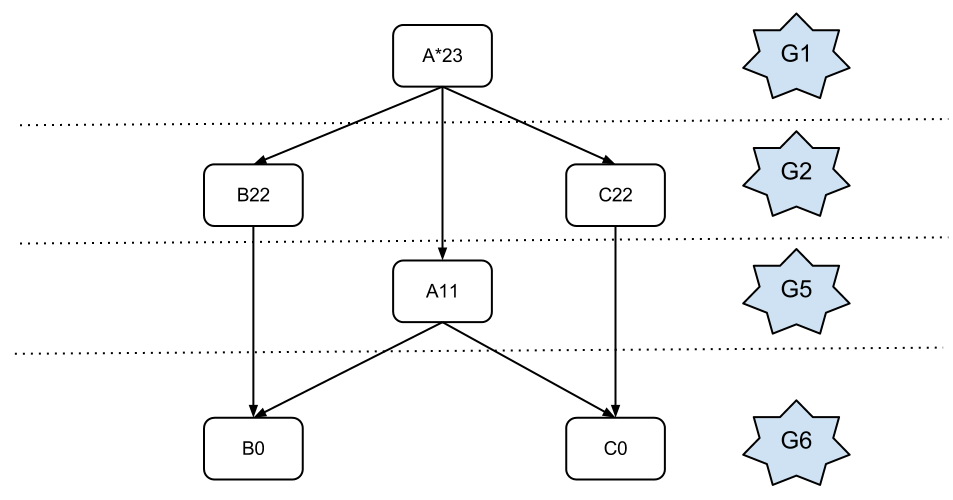}
  \caption{Specialization diagram for the $G_5, G_6$ solutions.}
  \label{fig:g56}
\end{figure}

\textbf{Remark:} the invariant analysis in Section \ref{sect:g3}
brings to light a minor classification mistake found in line 6 of
\cite[Table 24.2]{ES}. This line describes a $G_3$ class which is
listed as $\BL_{11}$ in our Table \ref{tab:011}.  Kundt-Ehlers give
the solution as $a u^{-2} \ln\zeta$ with $a$ a real constant.  This is
incorrect; the leading coefficient should be an arbitrary complex
number.

\section{Vacuum pp-wave spacetimes}
\label{sect:ppwaves}
% Note: general theory is well known.  No neeed to include -RM 12 March
% Following the work of \cite{Collins91} we choose the normalized dyad $ \{ o^A, \iota^A \}$ satisfying, $o_A \iota^A = 1 $. We define the generic symbol $\zeta_a^A$ for the dyad: 

% \beq & \zeta_0^A = o^A,~~\zeta_1^A = \iota^A,~~ \bar{\zeta}_{0'}^{A'} = \bar{o}^{A'},~~ \bar{\zeta}_{1'}^{A'} = \bar{\iota}^{A'} & \nonumber \eeq
Throughout, we use the four-dimensional Newman-Penrose
formalism \cite{PR} adapted to a complex, null-tetrad $(\ve_a) = (m^a, m^{*a}, \ell^a, n^a) = (\delta, \bardel, D,\Delta)$.  These vectors satisfy
\[ \ell_a n^a = 1, \quad m_a m^{*a} = 1, \] with all other
cross-products zero.  Equivalently, letting $\bth^1,\ldots,\bth^4$
denote the dual coframe, the metric is given by
\[g = 2\bth^1 \bth^2 - 2\bth^3 \bth^4.\]
The connection 1-form and the
the curvature 2-form 
are defined, respectively by
\begin{eqnarray}
  d\bth^a = \bom^a{}_b\wedge \bth^b,\quad \bom_{(ab)}=0\\
  \bOm^a{}_b = d\bom^a{}_b + \bom^a{}_c \wedge \bom^c{}_d.
\end{eqnarray}
The connection components are labeled by the 12 Newman-Penrose
scalars:
\begin{eqnarray}
  \label{eq:NPsc1}
  \qquad -\bom_{14}= \sigma\, \bth^1+\rho\, \bth^2+\tau\, \bth^3 +
  \kappa\, \bth^4;\\  
  \qquad \bom_{23} = \mu\, \bth^1+\lambda\,  \bth^2+\nu\, \bth^3
  +\pi\, \bth^4;\\ 
  \label{eq:NPsc3}
  -(\bom_{12}+\bom_{34})/2 = \beta\, \bth^1 +
  \alpha\,\bth^2+\gamma\,\bth^3+\epsilon\,\bth^4.
\end{eqnarray}
The curvature components are labelled by the Ricci
scalar $\Lambda=\bar{\Lambda}$, traceless Ricci components
$\Phi_{AB}=\bar{\Phi}_{BA},\; A,B=0,1,2$, and Weyl components $\Psi_C
,\; C=0,\ldots,4$:
\begin{eqnarray*}
%   \fl
%   -\bOm_{14}= \Phi_{01}(\bth^{12}-\bth^{34})+
%   \Phi_{02}\bth^{13}-\Phi_{00}\bth^{24} 
%   -\Psi_0\bth^{14} +\left(\Psi_2 + 2\Lambda\right)\bth^{23}
%   -\Psi_1 (\bth^{12}+\bth^{34})\\
  \small
 \bOm_{14}= \Phi_{01}(\bth^{34}-\bth^{12})-
  \Phi_{02}\bth^{13}+\Phi_{00}\bth^{24} 
  +\Psi_0\bth^{14} -\left(\Psi_2 + 2\Lambda\right)\bth^{23}
  +\Psi_1 (\bth^{12}+\bth^{34})\\
  \small
   \bOm_{23} = \Phi_{21}(\bth^{12}-\bth^{34}) +\Phi_{22}
  \bth^{13}-\Phi_{20} \bth^{24}+\Psi_4\bth^{23} - 
  (\Psi_2 + 2\Lambda)\bth^{14}  -
  \Psi_3(\bth^{12}+ \bth^{34})\\
  \small
 (\bOm_{12} + \bOm_{34})/2=  - \Phi_{12}\bth^{13} +
  \Phi_{10}\bth^{24}+\Psi_1\bth^{14}-
  \Psi_3\bth^{23}+\\ 
  \small
  \qquad +\Phi_{11}
  (\bth^{34}-\bth^{12})   +(\Psi_2-\Lambda)(\bth^{12}+\bth^{34}),
\end{eqnarray*}
where $\bth^{ab} = \bth^a\wedge\bth^b$.  

A pp-wave is a space-time admitting a covariantly constant null vector
field.  this entails
\[ \kappa=\sigma=\rho=\tau = 0.\] Such space-times are necessarily
Petrov type N or type O and belong to the Kundt class
\cite[Sect. 24.5]{ES}.  A vacuum pp-wave that isn't flat-space
is necessarily type N:
\[ \Phi_{AB'} = 0,\quad \Psi_0 = \Psi_1 = \Psi_2 = \Psi_3 = 0,\;
\Psi_4 \neq 0,\] Applying a boost and a spatial rotation we normalize
the tetrad by setting $\Psi_4 \to 1$.  Therefore, there are no 0th
order Cartan invariants.  The remaining frame freedom consists of the
2-dimensional group of null rotations.

The above constraints can be integrated to yield the following class
of exact solutions \cite[Section 24.5]{ES}:
\begin{equation}
  \label{eq:vsinvacmetric}
  \d s^2 = 2\d\zeta \d\bar\zeta -2\d u \d v- (f+\bar{f}) \d u^2,
\end{equation}
where $f=f(\zeta,u)$ is analytic in $\zeta$.  The above form is
preserved by the following class of transformations:
\begin{eqnarray}
  \label{eq:zetaxform}
  \hat\zeta &=& \e^{\i k}(\zeta+h(u))\\
  \label{eq:vxform}
  \hat{v} &=&  a(v+h'(u)\bar\zeta + \bar{h}'(u) \zeta+g(u))\\
  \label{eq:uxform}
  \hat{u} &=& (u+u_0)/a\\
  \label{eq:fxform}
  \hat{f}&=& a^2(f-\bar{h}''(u)\zeta+1/2(h'(u) \bar{h}'(u)-g(u)))
\end{eqnarray}

The Bianchi identities \cite[(7.32c) (7.32d)]{ES} impose:
\begin{equation}
\label{TNbianchi}   
\beta=\epsilon=0
\end{equation}
% General theory.  DOn't include -- RM 12 March
% The first covariant derivative is then defined by: \beq
% (D \Psi)_{\mu f'} = \Psi_{ABCD;EF'} \zeta_a^A \zeta_b^B \zeta_c^C
% \zeta_d^D \zeta_e^E \bar{\zeta}_{f'}^{F'} \nonumber \eeq \noindent
% where $\mu$ of the unprimed dyad vectors are $\zeta_1^A$'s. 
Using the notation of \cite{Collins91}, the non-vanishing 1st order
components are:  
\[ (D \Psi)_{50'} = 4 \alpha,\quad (D \Psi)_{51'} = 4 \gamma. \] The
transformation law for these components is \cite[(7.7c)]{ES}
\begin{equation}
  \label{eq:nrtransformation}
   \alpha' = \alpha,\quad \gamma' = \gamma +z \alpha,
\end{equation}
where $z$ is a
complex valued scalar.  Therefore, $\alpha$ is a 1st order Cartan
invariant and the invariant classification divides into two cases:
$\alpha=0$ and $\alpha\neq 0$. In the first case, $\gamma$
is an invariant, while in the 2nd case, we fix the tetrad by
normalizing $\gamma \to 0$.  We consider these two cases in more
detail.

%\subsection{The $\alpha\neq 0$ case.}
\begin{prop}
  \label{prop:alnonzero}
  Suppose that $\alpha\neq 0$. Then, $d_p=0$ for $p\geq 1$.  The
  possible values of the invariant count sequence are:
  \[(0,2,3,3),\, (0,1,3,3),\, (0,1,2,3,3),\, (0,2,2),\, (0,1,2,2),\,
  (0,1,1).\] The first 3 possibilities describe a $G_1$, the next 2
  possibilities are a $G_2$, and the last possibility is a $G_3$.  The
  Cartan invariants are generated by 
  \[ \bardel^n \alpha,\quad \delta^j \Delta^{n-j} \mu, \quad
  \Delta^{n} \nu,\quad 0\leq j\leq n,\; n=0,1,2,\ldots\] and their
  complex conjugates, where the above spin coefficients are calculated
  relative to the normalized $\Psi_4\to 1, \gamma\to 0$ tetrad.
\end{prop}

 \begin{prop}
   \label{prop:alzero}
   Suppose that $\alpha=0$.  Then $d_p = 2$ for all $p$.  The possible
   values of the invariant count sequence are
   \[ (0,1,1),\; (0,0).\] The first possibility describes a $G_5$. The
   second possibility describes a $G_6$ (homogeneous space).  The
   Cartan invariants are generated by
   \[\Delta^{n} \gamma,\quad n=0,1,2,\ldots\]
   and their complex conjugates, calculated relative to a tetrad
   normalized by $\Psi_4\to 1$.
 \end{prop}
\noindent
In the following sections we will show that each of these cases
describes a well-defined class of solutions, and go on to derive a the
canonical forms for the metric in each case.

We now turn to the proof of Proposition \ref{prop:alnonzero}, which
concerns the $\alpha\neq 0$ case.  The NP equations \cite[(7.21f)
(7.21o)]{ES} imply the additional constraints
\[ \pi= \lambda = 0.\] The non-vanishing 2nd order curvature
components are \cite[(4.2a)-(4.2t)]{Collins91}:
\begin{eqnarray*} 
  (D^2 \Psi)_{50';00'} &=& 4D \alpha ,\\
%  (D^2 \Psi)_{50';01'} &=& 4 \delta \alpha - 4 \baral \alpha , \\
  (D^2 \Psi )_{50';10'} &=& 4 \bardel \alpha + 20 \alpha^2 ,\\
  (D^2 \Psi)_{50';11'} &=& 4 \Delta \alpha  ,  \\ 
  (D^2\Psi)_{51';10'} &=& -4 \barmu \alpha ,\\
  (D^2 \Psi)_{51';11'} &=& - 4\barnu \alpha ,\\
  (D^2 \Psi)_{41';11'} &=& - \barnu \alpha.
\end{eqnarray*}
Therefore, the independent 2nd order Cartan invariants are $\mu, \nu,
\bardel \alpha$ and the corresponding complex conjugates.
The commutator relations are
\begin{eqnarray}
  \label{eq:alnzcom1}
  \Delta D - D\Delta  &=& 0,\\
  \label{eq:alnzcom2}
  \delta D-D\delta &=& \baral D,\\
  \label{eq:alnzcom3}
  \delta \Delta - \Delta \delta &=& -\barnu
  D-\baral\Delta+\mu\delta,\\  
  \label{eq:alnzcom4}
  \bardel\delta -\delta\bardel  &=&
  (\barmu-\mu)D-\baral\bardel  + \alpha \delta
\end{eqnarray}
The NP-equations imply the following relations amongst the invariants:
\begin{eqnarray} 
  \label{eq:Dalpha}
  D \alpha &=& 0,\\ 
  \label{eq:deltaalpha}
  \delta \alpha &=& \alpha  \baral,\\ 
  \label{eq:Deltaalpha}
  \Delta \alpha &=&-\barmu\alpha, \\
  \label{eq:Dmu} 
  D \mu &=& 0,\\ 
  \label{eq:bardeltamu}
  \bardel \mu &=&-\alpha \mu \\
  \label{eq:Dnu}
  D \nu &=& 0,\\ 
  \label{eq:bardelnu}
  \bardel \nu&=& 1-3\alpha \nu ,\\
  \label{eq:deltanu}
  \delta \nu&=&   -\baral\nu+\Delta \mu+\mu^2
\end{eqnarray}
Higher order relations follow in a straight-forward manner from these
and from the commutator relations.  Fixing $\Psi_4\to 1$ reduces the
isotropy to null rotation.  Fixing $\gamma\to 0$ eliminates this frame
freedom.  Therefore, the isotropy is trivial.  Equation
\eqref{eq:deltaalpha} implies that $\alpha$ is not constant.  All
invariants are annihilated by $D$. Therefore, there are either 3, 2,
or 1 independent Cartan invariants.  The conclusions of Proposition
\ref{prop:alnonzero} now follow directly from the Karlhede algorithm.

Next we present the proof of Proposition \ref{prop:alzero}, which
treats the $\alpha = 0$ class.  As was mentioned above, the 1st order
Cartan invariants are generated by
\[ (D \Psi)_{51'} = 4 \gamma \]
The  Newman-Penrose equations \cite[(7.21f) (7.21o)
(7.21r)]{ES} imply
 \begin{equation} \label{NPBcase} 
   D \gamma = 0,\;
   \delta \gamma = 0,\;
   \bardel\gamma = 0. 
 \end{equation}
 There is only one non-zero 2nd order curvature component, namely
 \begin{equation*}
   (D^2 \Psi)_{51';11'} = 4 \Delta \gamma + 20 \gamma^2 + 4
   \bar{\gamma} \gamma,
 \end{equation*}
 The operator transformation law for null rotations is
 \cite[(7.7a)]{ES}
 \[ D' = D,\; \delta' = \delta + B D,\; \Delta' = \Delta + B
 \bardel + \bar{B} \delta + B\bar{B} D.\] Therefore, by
 \eqref{NPBcase}, $\Delta^n\gamma$ is well-defined, despite the fact
 that no canonical choice of $\Delta$ exists and is invariant with
 respect to null rotations.  By \cite[(7.6a)-(7.6d)]{ES} all
 commutators are spanned by $\delta, \bardel, D$.  This implies
 that
 \[ \delta \Delta^n \gamma = \bardel \Delta^n \gamma = D\Delta^n
 \gamma = 0.\] Therefore there are two possibilities.  Either $\gamma$
 is a constant, in which case we have a homogeneous $G_6$; or $\gamma$
 is the unique independent invariant, in which case we have a $G_5$.
 This concludes the proof of Proposition \ref{prop:alzero}.

 \section{The $G_1$ solutions}
\label{sect:g1}
In this section, we derive solutions for certain key $G_1$
sub-classes.  We assume that $\alpha\neq0$ for the remainder of this
section.  The solutions are summarized in Table \ref{tab:g1} and
\ref{tab:g101}. In the tables, $F=F(z)$ is an analytic function;
$g=g(u)$ is complex-valued function of $u$; $h=h(u)$ is a real-valued
function of $u$; and $k$ is a real constant.  The meaning of $g_1,
g_2, h_1, h_2, k_1,k_2$ are analogous.

In the preceding section we established that $\alpha\neq 0$ solutions
admit an invariant tetrad characterized by the normalizations
\begin{equation}
  \label{eq:alnznormalization}
  \Psi_4= 1,\quad \gamma=0
\end{equation}
Let $\omega^1,\omega^2 = (\omega^1)^*, \omega^3
=(\omega^3)^*, \omega^4 =(\omega^4)^*$ denote the coframe dual
to $\delta, 
\bardel, \Delta,D$. 

We introduce the following key invariants.
\begin{align}
  \label{eq:Adef}
  A&:=\bardel\alpha/ \alpha^2,\\
  \label{eq:Bdef}
  B&:= \mu A-\barmu,\quad B_1 = \re B,\; B_2 = \im B,\\
  \label{eq:Mdef}
  M&:= \alpha\mu,\\
  \label{eq:Xdef}
  X&:= B/(A\barA-1) ,\quad X_1 = \re X,\; X_2 = \im X,\quad A A^*\neq 1,\\
  \label{eq:Ydef}
  Y &:= (3-A)\nu-1/\alpha  + (\Delta \mu+\mu^2)/\baral,\\
  \label{eq:hnudef}
  \hnu&:= \nu + X(\mu+2\barX)/\baral ,\quad A\barA\neq 1,\\
  \label{eq:Upsilondef}
  \hUpsilon&:=  \Delta(\hnu/\barX) - 2\hnu +1/\alpha - 4\i X X_2 / \baral,\\
  \label{eq:hDeltadef}
  \hDelta &:= \Delta + \hz^* \delta + \hz\bardel + \hz \hz^* D,\quad
  \hz := \barX/\alpha, \\
  \label{eq:tXdef}
  \tX  &:= \Delta \log M^*/(1-A^*),\quad A\neq 1,\\
  \label{eq:tnudef}
  \tnu&:= \nu + \tX^*(\tX + 2\mu-\barA \tX^*)/\baral,\quad  A\neq 1, \\
  \label{eq:tUpsilondef}
  \tUpsilon&:=  \Delta(\tnu/\tX^*) - 2\tnu +1/\alpha - 4\i \tX \tX_2 /
  \baral,\\ 
  \label{eq:tDeltadef}
  \tDelta &:= \Delta + \tz^* \delta + \tz\bardel + \tz \hz^* D,\quad
  \tz := \tilde{X}^*/\alpha.
 %  \label{eq:Ydef}
 %  Y &:= 2\nu-1/\alpha + (\Delta\mu + \mu^2)/\baral\\
 % \label{eq:tDeltadef}
 %  \tDelta &:= \Delta + \tz^* \delta + \tz \bardel + \tz \tz^*D,\quad
 %  \tz = \tX^*/\alpha
\end{align}
\noindent

Even though the $\gamma\to 0$ normalization is the most obvious way to
select an invariant tetrad, an equally useful normalization is
$\hDelta\alpha \to 0$.  The reason is that a Killing vector $V$
necessarily annihilates all invariants, and hence it will turn out to
be useful to work in a frame where $\hDelta$ is a linear combination of
Killing vectors.  

For a given vector field $V$ let us write
\[ V=V^1\delta + V^2 \bardel+ V^3 \Delta + V^4 D \]
where 
\[(V^1)^* =V^2, \; (V^3)^* = V^3, \; (V^4)^* = V^4.\] The following
proposition shows that if $AA^*\neq 1$, then the normalization
$\hDelta \alpha \to 0$ selects a well-defined invariant tetrad.
\begin{prop}
  \label{prop:hDelta}
  Suppose that $A A^*\neq 1$.  Then, every vector field that satisfies
  \begin{equation}
    \label{eq:cLValpha}
     \cL_V\alpha = \cL_V\baral =0,\quad V^3\neq 0.
   \end{equation}
   has the form $V=a\hDelta + b D,\; a\neq 0$.
  If $A A^*= 1$, but $B\neq 0$, then \eqref{eq:cLValpha} does not have
  a solution. If $A A^*=1$ and $B=0$, then there is a 1-parameter
  family of solutions to \eqref{eq:cLValpha}.
\end{prop}
\begin{proof}
  The null-rotation transformation law for $\Delta$ is  \cite[7.7
  (c)]{ES},
  \begin{equation}
    \label{eq:nrDelta}
    \hDelta = \Delta + \hz^* \delta + \hz \bardel +\hz\hz^* D.
  \end{equation}
  Hence, by \eqref{eq:Dalpha}-\eqref{eq:Deltaalpha} and \eqref{eq:Adef}
  we seek a scalar $\hz$ such that
  \begin{equation}
    \label{eq:hDeltasystem}
    \begin{pmatrix}
      A\alpha & \baral \\
      \alpha & A^*\baral \\
    \end{pmatrix}
    \begin{pmatrix}
      \hz \\ \hz^*
    \end{pmatrix}
    = 
    \begin{pmatrix}
      \barmu\\ \mu
    \end{pmatrix}
  \end{equation}
  If $AA^*\neq 1$, the  solution is
  \begin{equation}
    \label{eq:zhDelta}  
    \hz  = \frac{\barA\barmu-\mu}{(A A^*-1)\alpha} = \frac{\barX}{\alpha}.
  \end{equation}
  If $A A^* = 0$, then the system has rank 1.  In this case the system
  is consistent if and only if 
  \[
  \begin{vmatrix}
    A\alpha & \barmu \\ \alpha & \mu
  \end{vmatrix}
  = \alpha B = 0.\]
\end{proof}

Next, we establish some key relations for these invariants and
certain other scalars that will prove useful in our calculations.
\begin{prop}
  \label{prop:alnonzerorels}
  Suppose that $\alpha\neq 0$. If  the normalization
  \eqref{eq:alnznormalization} holds then
  \begin{align}
    \label{eq:alphaZrel}
    \alpha &= e^{a-\bara}/(Z_a)^*  =e^{a-\bara}   (a_\zeta)^*,\\
    \label{eq:muLrel}
    \mu &= e^{-a-\bara} L_u \\ %= - e^{-a-\bara} r_\zeta,\\
    \label{eq:MLrel}
    M &= (e^{-2a}/Z_a)^* L_u,\\
    \label{eq:nuLrel}
    \nu &= e^{-a-3\bara}\left( Z_{uu} +
      (\Phi_a/Z_a)^*\right)=e^{-a-3\bara}(
    Z_{uu} + (f_\zeta)^*),\\
    \label{eq:ALrel}
    A &= -1 - (L_a)^*,\\
    \label{eq:om1def} 
    \omega^1 &= (\baral)^{-1} da,\\
    \label{eq:om3def} 
    \omega^3 &= e^{a+\bara} du,\\
    % \omega^4 &= e^{-a-\bara}\left( (f+\bar{f} + rr^*) du + dv +
    %   \bar{r} d\zeta + rd\bzeta\right)\\
    \omega^4&= e^{-a-\bara} \left( (f+f^* +Z_u \barZ_u) du +
      dv-Z_u d\bzeta-Z_u^* d\zeta\right),
    \intertext{where}
    \label{eq:adef}
    a&:= \frac{1}{4} \log f_{\zeta\zeta},\quad a_\zeta \neq 0,\\
    % r&:= a_u/a_\zeta,\\
    \zeta &=: Z(a,u),\quad \bzeta =: \barZ(\bara,u) ,\\
    L &:= \log Z_a\\
    \Phi(a,u) &:= f(\zeta,u).\intertext{We also have}
    \label{eq:muXA}
    \mu &= X A^*+ X^*,\\
    \delta A &= 0.    \intertext{  Furthermore, if $Q=Q(a,u)$ then }
    \label{eq:deltaQ}
    \delta Q&= \baral \, Q_a,\\
    \label{eq:DeltaQ}
    \Delta Q&= e^{-a-\bara} Q_u.
  \end{align}
\end{prop}
\noindent We begin by deriving some a key classes of $G_1$ solutions;
\emph{all} the various solutions discussed in this paper are
subclasses of these general categories.
\begin{prop}
  \label{prop:genform} 
  Suppose that $\alpha\neq 0$. The following
  conditions are equivalent:
  \begin{gather}
    \label{eq:prcondition}
    \delta \barA\, \delta^2 \!M = \delta M\, \delta^2 \!\barA\\
    \label{eq:prsolution}
    f(\zeta,u) = g_1F(g_2 \zeta) + g_3 \zeta,
  \end{gather}
  where $F=F(z)$ is an analytic function such that $F'''(z) \neq 0$
  and where $g_i=g_i(u),\; i=1,2,3$ are complex-valued such that $g_1,
  g_2\neq 0$. Furthermore, $\delta M =0$ if and only if $g_1 =
  g_2^{-2}$; i.e.,
  \begin{equation}
    \label{eq:A*23}
    f(\zeta,u) = F(g\zeta)g^{-2} + g_3 \zeta.
  \end{equation}
  In addition $M=0$ if and only if $g_1=g_2=1$; i.e.,
  \begin{equation}
    \label{eq:A23}
    f(\zeta,u) = F(\zeta) + g\zeta.
  \end{equation}
\end{prop}
\begin{proof}
  Our first claim is that \eqref{eq:prsolution} is equivalent to the
  following conditions:
  \begin{align}  
    \label{eq:fzz}
    &f_{\zeta\zeta} = g_4F_1(g_2 \zeta) ,\quad g_4 = g_1 g_2^2,\,
    F_1(z)
    = F''(z),\\   \nonumber
    &a = F_2(g_2 \zeta) + g_5,\quad g_5 = \frac{1}{4}\log g_4,\; F_2(z)
    = \frac{1}{4} \log F_1(z)\\ \nonumber
    &Z = F_3(a-g_5)/g_2,\quad   F_3(F_2(z)) =z, \\ 
    \label{eq:LF4rel}
    &L = F_4(a-g_5) + g_6,\quad g_6 = -\log g_2,\; F_4(z) = \log
    F_3'(z),\\
    \label{eq:LaLucomb}
    & L_u+g_7 L_a=g_8 ,\quad g_7 = g_5'(u),\; g_8 = g_6'(u).
  \end{align}
  Note that since $\alpha\neq 0$, by \eqref{eq:alphaZrel}, we must
  have $L \neq0$.   We now consider two cases.  

  First, let us consider the case of $\delta M = 0$.  Note that in
  this case \eqref{eq:prcondition} holds trivially.  Also, in this
  case, $L_{ua}=0$, and hence without loss of generality, $g_5=0$.
  The case of $M=0$ is true if and only if $L_u=0$. Here $g_5=0$ and
  $g_1=g_2=1$.

  Let us now consider the generic case where $\delta M\neq 0$.  In this
  case, \eqref{eq:prcondition} can be restated as
  \[ \delta\left(\frac{\delta \barA}{\delta M}\right)  = 0.\]
  Observe that
  \begin{align*}
    &\frac{\delta (\barA/M)}{\delta (1/M)} = \barA - M
    \frac{\delta\barA}{\delta M}\\
    &  \delta\left(\frac{\delta (\barA/M)}{\delta (1/M)}\right) = -
    M \delta\left(\frac{\delta \barA}{\delta M}\right) 
  \end{align*}
  Hence, \eqref{eq:prcondition} is equivalent to
  \[ \delta\left(\frac{\delta (\barA/M)}{\delta (1/M)}\right)  =
  0.\]
  Next, we observe that
  \[ \frac{\delta (\barA/M)}{\delta (1/M)} = -1-L_a +\frac{L_u
    L_{aa}}{L_{au}}.\]
  Hence,
  \[ \bardel\left(\frac{\delta (\barA/M)}{\delta (1/M)} \right)=
  0\] Hence, by \eqref{eq:DeltaQ}, condition \eqref{eq:prcondition} is
  equivalent to 
  \begin{align*}
    &\frac{\delta(\barA/M)}{ \delta(1/M)} = g,\quad g=g(u)\\
    &\delta\left(\frac{1+\barA-g}{M}\right)  =  \delta \left(
      \frac{-L_a + g}{L_u}\right) \barZ_{\bara}  e^{2\bara} = 0
  \end{align*}
  The latter condition is equivalent to \eqref{eq:LaLucomb}.
\end{proof}

% \noindent
% Here is another essential property of the $\delta M = 0$ class of solutions
% \begin{prop}
%   \label{prop:tDelta}
%   Suppose that $\delta M=0,\; A\neq 1$. 
%   Then, up to a scale factor, $V=\tDelta$
%   is the unique real vector field such that $\cL_V\mu =0$ and
%   $V^3\neq 0$.
% \end{prop}
% \begin{proof}
%   We seek a $\tz$ so that
%   \[ \tDelta = \Delta + \tz^* \delta + \tz \bardel +\tz\tz^* D \]
%   satisfies $\tDelta M = 0$.  The necessary requirement is
%   \[ \Delta M + \tz \bardel M = \Delta M + \tz\alpha M(A-1) = 0.\]
%   This gives
%   \[   \tz = -\Delta M/(M\alpha(A-1)) \]
% \end{proof}

\begin{prop}
  \label{prop:B*23} Suppose that $B_1\neq 0$ and $A\barA \neq 1$. Then
  the following are equivalent: (i) $B_2/B_1=k$, is a real constant and
  (ii) $f(\zeta,u) = F(h^{\i k}\zeta)h^{2}+g\zeta$.
\end{prop}
\begin{proof}
  Let $C=1+\i k$ so that condition (i) is equivalent to
  \begin{equation}
    \label{eq:BCrel}
    \frac{B}{B^*} = \frac{B_1 + \i B_2}{B_1- \i B_2} =
    \frac{C}{C^*},
  \end{equation}
  or $\im(B/C)=0$.  Suppose that (i) holds. By \eqref{eq:Bdef},
  \begin{align*}
    A^* &=  (B^*+\mu)/\mu^*\\
    A A^*-1 &= A(B^*+\mu)/\mu^* -1 =(AB^*+A\mu-\mu^*)/\mu^*
    = (A B^*+B)/\mu^*\\
    &= (B/\mu^*)\left( 1+(B^*/B)  A\right). 
  \end{align*}
  Hence, assuming \eqref{eq:BCrel} and by \eqref{eq:muLrel} \eqref{eq:ALrel},
  \begin{gather}
    e^{-a-a^*}(C/B)(A\barA-1) = e^{-a-a^*}(C+\barC A)/\barmu =
    (-2\i k - CL_a)/L_u
  \end{gather}
  By Proposition \ref{prop:alnonzerorels}, the above is both real and
  holomorphic in $a$, and hence independent of $a$.  Hence,
  \begin{gather}
    \label{eq:B*23LuLarel}
    L_u  + (1+\i k)h_1 L_a = -2\i k h_1, 
  \end{gather}
  where $h_1=h_1(u)\neq 0$ is real.  Conversely, \eqref{eq:B*23LuLarel}
  with $h_1\neq 0$ implies condition (i).  Hence,
  \begin{gather*}
    L = F(a-(1+\i k) h_2) - 2\i k h_2,\quad h_2'(u) = h_1(u) \\
    Z = F(a-(1+\i k) h_2)\e^{-2\i k h_2}\\
    a = (1+\i k) h_2+ F(\e^{2\i k h_2}\zeta)\\
    f_{\zeta\zeta} = h^{2+2\i k} F(h^{\i k} \zeta),\quad h =\e^{2
      h_2}\\
    f = F(h^{\i k} \zeta) h^{2}+ g\zeta
  \end{gather*}
\end{proof}
\begin{prop}
  \label{prop:g1e} Suppose that  $A \neq 1$. Then,
  the following are equivalent: (i) $B_1=0$, and
  (ii) $f(\zeta,u) = F(\e^{\i h}\zeta)+g\zeta$.
\end{prop}
\begin{proof}
  By \eqref{eq:Bdef} \eqref{eq:muLrel}
  \eqref{eq:ALrel}, condition (i) is equivalent to
  \[ (L_a+2) \barL_u + L_u (L_a+2)^* = 0.\]
  where $L_a\neq -2$, by assumption.
  Hence, 
  \begin{gather*}
    L_u = \i h_1 (L_a+2),\quad L_u  - \i h_1 L_a = 2 \i h_1,
  \end{gather*}
  where $h_1=h_1(u)$ is real. Hence,
  \begin{gather*}
    L = F(a+\i h_2) + 2\i h_2,\quad h_2'(u) = h_1(u) \\
    Z = F(a+\i h_2)\e^{2\i h_2}\\
    a = -\i h_2+ F(\e^{-2\i h_2}\zeta)\\
    f_{\zeta\zeta} = h^{2\i h} F(\e^{\i h} \zeta),\quad h =-h_2/2\\
    f = F(\e^{\i h} \zeta) + g\zeta
  \end{gather*}
\end{proof}

\begin{prop}
  \label{prop:aa*}
  The following are equivalent: (i) $A A^*=1$, and (ii) $L= Pa+g$ whereb
  $g=g(u)$ and $P=P(u)$ such that $P P^*+P+P^*=0$.  
\end{prop}
\begin{proof}
  By \eqref{eq:ALrel}, $A^*$ is holomorphic in $a$.  Hence, if $A A^*
  =1$ then, $A$ must be independent of $a$; i.e., $L_a = P =P(u)$.
  Since $A^*=1/A=-1-L_a$,  condition (ii) follows.
\end{proof}

\begin{prop}
  \label{prop:P23}
  Suppose that $A^2\neq 1$.  The following are equivalent: (i)
  $AA^*=1$ and (ii) $f(\zeta,u) = (\e^{g_1}\zeta)^{\i h} + g_2\zeta$
\end{prop}
\begin{proof}
  By  Proposition \ref{prop:aa*}, condition (i) is equivalent to
  \[ L = P a - g,\quad g =g(u) \]
  where, by assumption, $P\neq 0, -2$.  Hence $\re(1/P) =-1/2$, whence
  \[ P= -2\i /(h+\i ).\]
  where $h=h(u)\neq 0$ is real.  Hence,
  \begin{gather*}
    Z = \exp(P a-g),\quad     a = (\log(\zeta)+g)/P,\quad
    f_{\zeta\zeta} = (\e^g\zeta)^{4/P}=(\e^g\zeta)^{-2+2\i h}
  \end{gather*}
  Hence (ii) follows with 
  \[ g= (1+P/2) g_1 + \text{const.}\]
\end{proof}
\begin{prop}
  \label{prop:E23}
  The following are equivalent: (i) $A=-1$ and (ii) $f(\zeta,u) =
  \exp(g_1\zeta) + g_2\zeta$
\end{prop}
\begin{proof}
  By the Lemma, condition (i) can be restated as $L = g$. Hence,
  \begin{gather*}
    Z = ga,\quad a = g\zeta,\quad f_{\zeta\zeta} = \exp(\e^g\zeta),\quad
    f = \exp(g_1\zeta) + g_2\zeta
  \end{gather*}
\end{proof}
\begin{prop}
  \label{prop:L23}
  The following are equivalent: (i) $A=1$ and (ii) $f(\zeta,u) =
  g_1\log\zeta + g_2\zeta$.  
\end{prop}
\begin{proof}
  By the Lemma, condition (i) can be restated as $L = -2a+g$.  Hence,
  \begin{gather*}
    Z = \exp(-2a+g),\quad -2a = \log(\zeta)-g,\quad  f_{\zeta\zeta} =
    \e^{2g}\zeta^{-2},\quad 
    f = g_1\log\zeta + g_2\zeta
  \end{gather*}
\end{proof}
\noindent
Note that if $A=1$, then $B=\mu-\mu^*$.  Hence, if $A=1$, then $B_1=0$
automatically.

\section{The $(0,1)$ class}
\label{sect:01}
Above we showed that $\alpha,\baral$ generate the 1st order
invariants.  Generically, these are independent and hence,
generically, the invariant count is $(0,2)$.  However, an important
subclass occurs for which $\d\alpha\wedge \d\baral = 0$.  We will
refer to these as the $(0,1)$ solutions. The next two Propositions
characterize the (0,1) solutions in terms of invariants.

\begin{prop}
  \label{prop:albaldep}
  If $\mu\neq 0$, then $\d\alpha\wedge \d\baral = 0$ if and only if
  $B=0$.  In this case, the condition $A A^*=1$ follows
  automatically. If $\mu=0$, then $\d\alpha\wedge \d\baral = 0$ if and
  only if $A\barA =1$. 
\end{prop}
\begin{proof}
  By \eqref{eq:deltaalpha} \eqref{eq:Deltaalpha},
  \begin{align*}
    &\delta \alpha \bardel \baral - \bardel \alpha
    \delta \baral= \alpha^2 \baral^2(A\barA-1)\\
    &\delta \alpha \Delta \baral - \delta \baral \Delta
    \alpha =\alpha\baral^2 B
  \end{align*}
  Hence, the condition $\d\alpha \wedge \d\baral = 0$ is equivalent
  to the conjunction of $A\barA=1$ and $B=0$.  However, if $\mu
  \neq 0$ and $B=0$, then 
  $A = \barmu/\mu$, and hence $A\barA =1$ automatically.  Therefore,
  if $\mu\neq 0$, then the condition $B=0$ suffices.  On the other hand,
  if $\mu=0$, then $B=0$, and therefore the condition $A\barA = 1$ suffices.
\end{proof}

\begin{prop}
  \label{prop:albaldep1}
  Suppose that $B=0$ and $A\barA=1$. Then, necessarily $A$ is a constant
  and $\delta M = 0$.
\end{prop}
\begin{proof}
  By Proposition \ref{prop:aa*},
  $L=Pa+g$ where $P=P(u), g=g(u)$.  Since $B=0$, we have
  \begin{equation}
    \label{eq:LuLarel}
    L_u+\barL_u (1+L_a) = 0.
  \end{equation}
  Taking the derivative with respect to $a$ gives $L_{ua}=0$.  Hence,
  $A$ must be a constant.  Furthermore, by \eqref{eq:MLrel}
  \eqref{eq:deltaQ},
  \[ \delta(M) = \frac{e^{-a-\bara}}{Z_a \barZ_{\bara}} L_{au} =
  0,\]
  as was to be shown.
\end{proof}

\begin{lem}
  \label{lem:aa*b0}
  Suppose that  $B=0$ and $AA^*=1$.  
  If $A\neq 1$, then
  \begin{equation}
    \label{eq:LArel}
    L = -\frac{A+1}{A}\, a + \frac{A-1}{A}\, (k+\i h)
  \end{equation}
  where $k$ is a real constant, and $h=h(u)$ is real.
  If $A=1$, then
  \begin{equation}
    \label{eq:LArel1}
    L = -2a + h+k\i.
  \end{equation}
\end{lem}
\begin{proof}
  By Propositions \ref{prop:aa*}, \ref{prop:albaldep1}, $L=Pa+g$ where
  $g=g(u)$ and $P=-(A+1)/A$ is a constant.  Hence, equation
  \eqref{eq:LuLarel} can be restated as
  \[ g'(u) -  (g'(u))^*/A = 0.\]
  If $A\neq 1$, we
  multiply both sides by $A/(A-1)$ to obtain $\re(A/(A-1)g'(u)) = 0$.
  This gives us \eqref{eq:LArel}.  If $A=1$, then \eqref{eq:LArel1}
  follows immediately.
\end{proof}

\begin{prop}
  \label{prop:013}
  A type $(0,1)$ solution belongs to one of the classes shown in Table
  \ref{tab:g101}.
\end{prop}
\begin{proof}  By Proposition \ref{prop:albaldep}, $B=0$ and $AA^*=1$.
  We proceed by cases. Suppose that $A^2\neq 1$.  By Proposition
  \ref{prop:P23}, 
  \[f=(\e^{g_1} \zeta)^{2\i k} + g_2 \zeta,\quad L=Pa-(1+P/2) g_1.\]
  Since $(1+P/2)= (A-1)/(2A)$, we must have $g_1= k+\i h$ by Lemma
  \ref{lem:aa*b0} . This gives form $\rP_{13}$.  Next, consider the
  case $A=-1$.  Here, $L=k+\i h$.  By Proposition \ref{prop:E23} we
  arrive at form $\rE_{13}$.  Finally, if $A=1$, then
  \eqref{eq:LArel1} and Proposition \ref{prop:L23} give form
  $\rL_{13}$.
\end{proof}

\section{The $G_2$ precursors}
\label{sect:g2precursor}
As above, we assume that $\alpha\neq 0$ and that
$\delta,\delta^*,\Delta,D$ is a tetrad normalized so that $\Psi_4\to
1$ and $\gamma\to 0$.  In this section we classify the solutions that
satisfy the following definition.
\begin{defn}
  \label{def:g2pre}
  We say that a vacuum pp-wave metric is a $G_2$-precursor if there
  exists a vector field $V=V^1\delta + V^2 \bardel + V^3 \Delta + V^4
  D$ such that
  \begin{equation}
    \label{eq:precursor}
    \cL_V \omega^1 = \cL_V \omega^3 = 0,\quad  V^1\neq 0,\text{ or }
    V^3\neq 0
  \end{equation}
\end{defn}
\noindent
A Killing vector annihilates all invariant scalars and invariant
differential forms \cite[Ch. 8-10]{olver}.  Thus, the ``precursor''
terminology reflects the fact that \eqref{eq:precursor} is a necessary,
but not sufficient condition, for the existence of a Killing vector
independent from $D=\partial_v$.  The requisite propositions and
proofs are presented below. The resulting classification of precursor
solutions is summarized in Tables \ref{tab:preg202} and
\ref{tab:preg201}.

\begin{prop}
  \label{prop:om13}
  Let $V=V^1\delta + V^2 \bardel + V^3 \Delta + V^4 D$ be a vector
  field. Relation $\cL_V\omega^1=\cL_V\omega^3=0$ holds if and only if
  $C=\baral V^1$ is a constant, while $V^3$ satisfies
  \begin{align}
    \label{eq:V3rel}
    V^3\barmu &= C + \barC A\\
    \label{eq:deltaV3}
    \delta V^3 &= \baral V^3,\\
    \label{eq:DeltaV3}
    \Delta V^3 &= -C-\barC.
  \end{align}
\end{prop}
\begin{proof}
  By \eqref{eq:Dalpha} - \eqref{eq:Deltaalpha} and \eqref{eq:Adef},
  \begin{align}
    \label{eq:valpharel} \cL_V \alpha &= \alpha (C + \barC A - V^3
    \barmu).\\
    \intertext{By \eqref{eq:om1def} and the definition of $C$,}
    \label{eq:LVal*om1}
    \cL_V(\baral\omega^1) &= \cL_V da = d(\cL_V a) = d(V\rfloor
    da) = d C,\\
    \intertext{We also have  the following identity:}
    \cL_V \omega^3 &= d(V\rfloor \omega^3) + V\rfloor \d\omega^3 \\
    &= d(V^3) + (\baral V^1+ \alpha V^2)\omega^3 - \baral V^3 \omega^1
    - \alpha V^3 \omega^2\\
    \label{eq:LVomega3}
    &= (\delta V^3 - \baral V^3) \omega^1 + (\bardel V^3 -
    \alpha V^3) \omega^2 + (\Delta V^3 +C+\barC)\omega^3
  \end{align} 
  The desired equivalence follows immediately.
\end{proof}

\begin{prop}
  \label{prop:precursor}
  If $B \neq 0$, then \eqref{eq:precursor} is equivalent to the 
  the conjunction of
  \begin{equation}
    \label{eq:albaralmu}
    \d\alpha\wedge\d\baral\wedge\d\mu = 0,
  \end{equation}
  and the condition
  \begin{equation}
    \label{eq:dBB*}
    \d (B/B^*) = 0.
  \end{equation}
\end{prop}
\begin{proof}
  Note the following structure equations, which are dual to the
  commutator relations \eqref{eq:alnzcom1}-\eqref{eq:alnzcom4}
  \begin{align}
    \label{eq:dom1}
    \d\omega^1 &= \alpha \omega^1\wedge \omega^2 -
    \mu\omega^1\wedge\omega^3,\\
    \label{eq:dom3}
    \d\omega^3 &= (\baral\omega^1+\alpha\omega^2)\wedge \omega^3 ,\\
    \label{eq:dom4}
    \d\omega^4 &= (\barmu-\mu) \omega^1\wedge \omega^2  + (\barnu\omega^1
    +\nu\omega^2) \wedge\omega^3 - (\baral
    \omega^1+\alpha\omega^2)\wedge \omega^4.
  \end{align}
  If \eqref{eq:precursor} holds, then 
  \[ \cL_V\alpha = \cL_V\baral = \cL_V \mu = 0,\] because
  $\alpha,\baral,\mu$ are the structure functions in \eqref{eq:dom1}
  \eqref{eq:dom3}.  But, if 3 functions on a 4-dimensional manifold
  are annihilated by 2 independent vector fields, then they must be
  functionally dependent.  Therefore \eqref{eq:albaralmu} holds.  By
  Proposition \ref{prop:hDelta}, 
  \[ V = a \hDelta + b D ,\] where $\hDelta$ is defined as per
  \eqref{eq:hDeltadef}, and $a,b$ are some functions.  By Proposition
  \ref{prop:om13},
  \[ C=\baral   V^1 = a X,\quad 
   C^*=\al V^2 = a X^* \] are constants.  
  Hence, by \eqref{eq:Xdef},
  \begin{equation}
    \label{eq:XBC}
    B/B^* = X/X^* = C/C^*
  \end{equation}
  is a constant.

  Conversely, suppose that \eqref{eq:albaralmu} and \eqref{eq:dBB*}
  hold.  By assumption, \eqref{eq:XBC} holds for some constant $C$.
  Hence, $C/X = C^*/X^*$ is real.  Set $V= C/X \hDelta$.  This is a
  real vector field such that, by construction, $ \cL_V\alpha =
  \cL_V\al^* = 0$.  Since $\alpha,\alpha^*,\mu$ are functionally
  dependent, we also have $\cL_V \mu = 0$.  By relation
  \eqref{eq:LVal*om1}, $\cL_V \omega^1 = 0$.  By \eqref{eq:dom1} and
  \eqref{eq:LVomega3},
  \begin{gather*}
    0=\cL_V d\omega^1  = -(\cL_V\mu) \omega^1\wedge \omega^3 - \mu
    \omega^1 \wedge \cL_V\omega^3 = -\mu \omega^1\wedge \cL_V\omega^3
  \end{gather*}
  Since $\cL_V \omega^3$ is real and $\mu\neq 0$ by assumption, it
  follows that $\cL_V\omega^3=0$.
\end{proof}
\noindent
\textbf{Remark:}
Observe that
\[ \frac{B}{B^*} = \frac{B_1 + \i B_2}{B_1-\i B_2} = \frac{1+ \i
  B_2/B_1}{1-\i B_2/B_1}.\]
Hence, if $B_1\neq 0$, then condition \eqref{eq:dBB*} can be
conveniently expressed as $B_2/B_1=k$ where $k$ is a real constant.

We now show that type $(0,2)$ precursor solutions belong to the 4
classes shown in Table \ref{tab:preg202}.  Proposition
\ref{prop:Lprecursor} characterize the precursor solutions for which
$V^3=0$.  Proposition \ref{prop:Aprecursor} characterizes precursor
solutions for which $V^1=0$.  This leaves the case where both $V^1,
V^3$ are non-zero.  Since we are considering type $(0,2)$ solutions,
we exclude the possibility that $B=0$.  The possibility that $B\neq 0$
but $AA^*=1$ is excluded by Proposition \ref{prop:hDelta}. The
remaining possibilities can be divided into the case $B_1\neq 0$ and
the case $B_1=0$.  Proposition \ref{prop:B23} deals with the former
and \ref{prop:C23} with the latter.

\begin{prop}
  \label{prop:Lprecursor}
  There exists a vector field $V$ such that
  \begin{equation}
    \label{eq:precursorV30}
    \cL_V\omega^1=\cL_V\omega^3=0,\quad     V^1\neq0, V^3= 0
  \end{equation}
  if and only if  $A=1$.
\end{prop}
\begin{proof}
  Suppose that \eqref{eq:precursorV30} holds.  By \eqref{eq:DeltaV3},
  $C+\barC = 0$, and hence $C= \baral V^1$ is imaginary.  Hence, by
  \eqref{eq:V3rel}, $C+C^* A = 0$, which means that $A=1$.
  Conversely, if $A=1$, then in order for \eqref{eq:V3rel} -
  \eqref{eq:DeltaV3} to hold, it suffices to set $V^1= \i/\baral$,
  $V^3=0$.
\end{proof}

\begin{prop}
  \label{prop:Aprecursor}
  There exists a vector field $V$ such that
  \begin{equation}
    \label{eq:precursorV10}
    \cL_V\omega^1=\cL_V\omega^3=0,\quad     V^1= 0, V^3\neq 0
  \end{equation}
  if and only if   $\mu=0$.
\end{prop}
\begin{proof}
  Suppose that \eqref{eq:precursorV10} holds. Hence, by \eqref{eq:valpharel} ,
  \[ \cL_V\alpha = -V^3\alpha\mu = 0.\]
  Therefore, $\mu=0$.  To prove the converse, it suffices to
  take $V^3=\e^{a+\bara}$.  Relations \eqref{eq:deltaV3} and
  \eqref{eq:DeltaV3} follow  by \eqref{eq:deltaQ}
  \eqref{eq:DeltaQ}.
\end{proof}

\begin{prop}
  \label{prop:B23}
  Suppose $B_1\neq 0, A A^*\neq 1$. The following are equivalent: (i)
  condition \eqref{eq:precursor} holds; (ii) $B_2/B_1=k,\; \Delta X_1 =
  2X_1^2$; (iii) $f(\zeta,u) = F(u^{-\i k} \zeta)u^{-2}+ g \zeta$.
\end{prop}
\begin{proof}
  Suppose that (i) holds. Since $V^3$ is real, by Proposition
  \ref{prop:om13},
  \[ \barmu(C+\barC \barA) - \mu(\barC + CA) = C \barB - \barC B =
  0.\] Since $B\neq 0$, we have $\mu\neq 0$ also.  Hence, $C\neq0$, by
  \eqref{eq:V3rel}.  Hence,
  \[\frac{C}{C^*}=\frac{B}{B^*} = \frac{1+\i B_2/B_1}{1-\i B_2/B_1}
  .\] Hence, $C=1+\i k$, without loss of generality, and $B_2/B_1=k$.
  Furthermore, since $X/X^*=B/B^*$, we have
  \begin{gather}
    \label{eq:XArel}
    \frac{1}{X_1} = \frac{C}{X} = \frac{C(B+\barB A)}{B\barmu} =
    \frac{C+\barC A}{\barmu} = V^3
  \end{gather}
  Therefore, (ii) follows by \eqref{eq:DeltaV3}.

  Next, we show that (ii) implies (iii).  By Proposition
  \ref{prop:B*23}, $f(\zeta,u)= F(h^{\i k}\zeta)h^2+g\zeta$ belongs to
  class $B^*_{23}$.  In the proof of Proposition \ref{prop:B*23}, we
  showed that
  \begin{gather*}
    e^{-a-\bara}/ X_1 = 1/h_1, 
  \end{gather*}
  where $h_1=h_1(u)$ is real. Hence, by \eqref{eq:DeltaQ}
  \begin{gather*}
    \Delta (1/X_1) =  \Delta (e^{a+\bara}/h_1) =
    (1/h_1)'(u),\\
    (1/h_1)'(u) + 2  = 0,\\
    h_1 = -1/(2u).
  \end{gather*}
  In the last step we can omit the constant of integration because of
  transformation freedom
  \eqref{eq:uxform}. Therefore
  \[ L_u -\left(\frac{1+\i k}{2u}\right) L_a=\frac{\i k}{u}.  \]
  Following the steps in the proof of Proposition \ref{prop:B*23} gives
  $h=u^{-1}$, which specializes solution form $B^*_{23}$ to form
  $B_{23}$.

  Finally we show that (iii) implies (i). For this, it suffices to set
  $V^1= C/\baral$ where $C=1+\i k$ and to set
  \[ V^3=1/X_1 =-2ue^{a+\bara}\] 
  Conditions \eqref{eq:deltaV3} \eqref{eq:DeltaV3} follow by
  \eqref{eq:deltaQ} \eqref{eq:DeltaQ}.
\end{proof}

\begin{prop}
  \label{prop:C23} Suppose that $B_1= 0, \mu\neq 0, AA^*\neq 1$. The
  following are equivalent: (i) condition \eqref{eq:precursor} holds;
  (ii) $\Delta X_2 = 0$; (iii) $f(\zeta,u) = F(\e^{\i u} \zeta)+ g
  \zeta$.
\end{prop}
\begin{proof}
  Let us show that (i) implies (ii).  As above, $C=\baral V^1\neq 0$
  is a constant such that $\im(B/C)= 0$. Since $B_1=0$ we have $C=\i$
  without loss of generality.  Hence, $V^3=1/X_2$ and $\Delta X_2 = 0$
  by \eqref{eq:DeltaV3}. 

  Next, we show that (ii) implies (iii).  By assumption, $f(\zeta,u) =
  F(\e^{\i h}\zeta) + g\zeta$ belongs to class $C^*_{23}$. Since $B=\i
  B_2$ we have by
  \[ -\i/X_2 = 1/\barX=(\barB+B\barA)/(\barB\mu) = (1-\barA)/\mu.\]
  Hence, by   Proposition \ref{prop:g1e}
  \[ \e^{-a-\bara}/ X_2 = -\i(2+L_a)/L_u=-1/h_1\] where
  $h_1(u)=-2h'(u)\neq 0$ is real. Since $\Delta X_2=0$, we infer that
  $h_1$ is a constant. Hence, without loss of generality, $h(u) = u$.

  Finally we show that (iii) implies (i). For this, it suffices to set
  $V^1= C/\baral$ where $C=\i$ and to set
  \[ V^3=1/X_2 =-2ke^{a+\bara}\] 
  Conditions \eqref{eq:deltaV3} \eqref{eq:DeltaV3} follow by
  \eqref{eq:deltaQ} \eqref{eq:DeltaQ}.
\end{proof}

We now classify the type $(0,1)$ precursor solutions.

% \begin{lem}
%   \label{lem:tDeltaB0}
%   Suppose that $B=0,\; AA^*=1,\; A\neq 1$. 
%   Then,  $\tDelta\alpha =0$.
% \end{lem}
% \begin{proof}
%   Write
%   \[  \tz = -\Delta \log(M)/(\alpha(A-1)).\]
%   Since $A\mu = \mu^*, \; A^* = 1/A,\; dA=0$, we have
%   \begin{align*}
%     M^* &=  \baral \barmu = AM  \baral/\alpha,\\
%     \tz^* &= -\frac{\Delta \log(M) + \Delta \log\baral- \Delta \log
%       \alpha}{\baral(1/A-1)} \\
%     &= -\frac{A(\Delta \log M +(A-1)\mu)}{\baral (1-A)}\\
%     \tDelta \alpha &= -\barmu \alpha + \tz A \alpha^2 + \tz^*\alpha
%     \baral\\
%     &= -AM-\frac{A\alpha \Delta \log M }{A-1} - \frac{\alpha A(\Delta \log M
%       +(A-1)\mu)}{ (1-A)} = 0
%   \end{align*}
% \end{proof}

\begin{prop}
  \label{prop:Delta1mu}
  Suppose that $B=0, A\neq 1, \mu\neq 0$.  Then \eqref{eq:precursor}
  holds if and only if
  \begin{equation}
    \label{eq:Delta1mu}
    \Delta^2(1/\mu)=0
  \end{equation}
\end{prop}
\begin{proof}
  Suppose that \eqref{eq:precursor} holds.  By Proposition
  \ref{prop:albaldep1},  $A$ is a constant.  Hence,
  \eqref{eq:Delta1mu} follows by \eqref{eq:V3rel}
  \eqref{eq:DeltaV3}. Conversely, suppose that \eqref{eq:Delta1mu}
  holds.  
  % Our assumptions imply that equation  \eqref{eq:LArel} holds.
  % Hence by \eqref{eq:muLrel}
  By Proposition \ref{prop:om13}, we seek a constant $C$
  such that 
  \[V^3 = (C^*+C A^*)/\mu = (C+C^* A)/\barmu,\] and such that the
  above $V^3$ satisfies \eqref{eq:deltaV3} and \eqref{eq:DeltaV3}.  
  First, observe that $A^* = 1/A$ and $\barmu=A \mu$.  Hence,
  \[ (C+C^* A)/\barmu= (C/A+C^*)/\mu = (C^* + CA^*)/\mu.\] Therefore,
  $V^3$ is well-defined for any choice of $C$.  By Proposition
  \ref{prop:albaldep1}, $\delta(\alpha\mu) = 0$.  Hence
  \[ \alpha(\baral \mu + \delta\mu) = 0\quad \delta(1/\mu) =
  -\delta\mu/\mu^2 =\baral/\mu.\] Hence, \eqref{eq:deltaV3} is
  satisfied for all choices of $C$.  
  We now turn to condition \eqref{eq:DeltaV3}.  
By \eqref{eq:muLrel}
  and \eqref{eq:LArel} of Lemma \ref{lem:aa*b0}
  \[ \i(A-1)/A /\mu = \i(A-1)/A e^{a+\bara}/L_u = -\e^{a+\bara}/
  h'(u),\] where $h=h(u)$ is real. Hence, by \eqref{eq:DeltaQ},
  \begin{equation}
    \i (A-1)/A \Delta(1/\mu) = -h''(u)/(h'(u))^2= k
  \end{equation}
  is a real constant.  If $k=0$, then condition \eqref{eq:DeltaV3} can
  be satisfied by taking $C=\i$.  If $k\neq 0$, \eqref{eq:DeltaV3} is
  satisfied by taking $C= A/(A-1) + \i/k$.  With this choice,
  \begin{gather*}
    C^*+C A^* = \frac{1}{1-A} - \frac{\i}{k} +\left(\frac{A}{A-1}+
      \frac{\i}{k}\right) \frac{1}{A} = \frac{-\i(A-1)}{kA},\\
    \Delta V^3= -1,\\
    C+C^* = \frac{A}{A-1} + \frac{1}{1-A} = 1    
  \end{gather*}
\end{proof}

\begin{prop}
  \label{prop:13precursor} The type (0,1) precursor solutions belong
  to one of the classes shown in Table \ref{tab:preg201}.
\end{prop}
\begin{proof}
  By Proposition \ref{prop:Lprecursor} the $B=0, A=1$ solutions are
  automatically precursor solutions with $V^3=0, V^1\neq 0$.  We now
  classify all precursor solutions that admit a vector field that
  satisfies \eqref{eq:precursor} with $V^3\neq 0$.  We consider two
  cases: $\mu\neq 0$ and $\mu=0$.  Suppose the former.  By the above
  Lemma, a precursor solution is characterized by the condition
  $\Delta^2 (1/\mu) = 0$, which is equivalent to
  \begin{equation}
    \label{eq:hk}
    h''(u) +k \, h'(u)^2 = 0.
  \end{equation}
  where $h$ is the parameter in solution forms $\rP_{13}$, $\rE_{13}$,
  $\rL_{13}$.  This gives us four classes of solutions.  Class
  $\BP_{13}$ corresponds to the case $A\neq -1$ and $k\neq 0$.  In
  this case, the solution of \eqref{eq:hk}, without loss of
  generality, is $h= \frac{1}{k}\log u$. Class $\CP_{13}$ corresponds
  to $A\neq -1$ and $k=0$.  Here, without loss of generality, the
  solution to \eqref{eq:hk} is $h=u$.  Similarly, the condition $A=-1$
  gives solution classes $\BE_{13}$ and $\CE_{13}$.  Finally, consider
  the case of $A=1$.  Here $\mu^*=\mu$.  Hence, by Proposition
  \ref{prop:L23} $L = -2a+\i k + h$ where $h$ is real.  By Lemma
  \ref{prop:om13} we require that
  \[ V^3 = (C+C^*)/\mu \neq 0,\quad \Delta V^3 = (C+C^*)\Delta(1/\mu)
  = -(C+C^*).\] Since $\delta(\alpha\mu) =0$, we automatically have
  $\delta(1/\mu) = \alpha^* \mu$; condition \eqref{eq:deltaV3} is
  automatically satisfied.  Hence, a necessary and sufficient
  condition for a precursor solution is $\Delta(1/\mu) = -1$, or
  equivalently $\Delta\mu = \mu^2$.  This is equivalent to $h''(u) =
  h'(u)^2$, which, by employing the freedom \eqref{eq:uxform}, gives us
  $h(u) = -\log u$.  Employing the integration steps in Proposition
  \ref{prop:L23}, this gives us $f = C u^{-2} \log\zeta +
  g\zeta$, which is solution form $\BL_{13}$.

  Next, suppose that $\mu=0, A A^*=1$.  Here it suffice to specialize
  one of the Table \ref{tab:g101} solutions.  For classes $\rP_{13}$
  and $\rE_{13}$ we set $h \to 0$.  For the logarithmic solution
  $\rL_{13}$ we set $h\to k$, where the latter is a constant.
\end{proof}

\section{The $G_2$ solutions}
\label{sect:g2}
In this section we characterize and classify the vacuum pp-waves with
two independent Killing vectors.  Since a Killing vector annihilates
the invariant 1-forms $\omega^1,\ldots, \omega^4$, every $G_2$
solution is a specialization of the precursor metrics discussed in the
preceding Section.

We first present the invariant characterization of the generic, type
$(0,2,2)$ solutions, and then present the characterization of the type
$(0,1,2,2)$ solutions.  We then pass to a detailed classification, the
results of which are displayed in Tables \ref{tab:022} and
\ref{tab:0122}.

\begin{prop}
  \label{prop:022}
  A type (0,2,2) $G_2$ solution is characterized by
  \eqref{eq:precursor} and
  \begin{equation}
  \label{eq:022condition}
  d\alpha\wedge \d\baral\neq 0,\quad
  d\alpha\wedge\d\baral\wedge\d\nu = 0.
  \end{equation}
\end{prop}
\begin{proof}
  By Proposition \ref{prop:alnonzero}, the 2nd order Cartan
  invariants are generated by $A, \mu,\nu$.  Suppose that there exists
  a Killing vector $V$ independent from $D$.  Condition
  \eqref{eq:precursor} follows by definition.  Since Killing vectors
  annihilate invariants, there are at most two functionally
  independent invariants.  Hence, \eqref{eq:022condition} must hold.

  Conversely, suppose that \eqref{eq:precursor} and
  \eqref{eq:022condition} hold.  Dependence of $\mu$ follows by
  Proposition \ref{prop:precursor}. Furthermore,
  \[ \cL_V\alpha = \cL_V\alpha^*=0,\quad \cL_V \d\alpha = \d
  \cL_V\alpha = 0,\] where $V$ is the vector field in
  \eqref{eq:precursor}. By \eqref{eq:Dalpha} -- \eqref{eq:Deltaalpha}
  and \eqref{eq:Adef},
  \[      \d\alpha = \alpha(\baral \omega^1 + A\alpha \omega^2 -\mu
  \omega^3) \]
  Hence $\cL_V A = 0$.  Therefore, the invariant count is $(0,2,2)$.
\end{proof}

The type $(0,1,2,2)$ solutions split into two branches, depending on
whether or not $\mu$ is independent of $\alpha$.  We consider each
branch in turn.
\begin{prop}
  \label{prop:01nu}
  Suppose that $\d\alpha\wedge\d\baral = 0$ but that
  $\d\alpha\wedge\d\mu\neq 0$. Then a $G_2$ solution is characterized
  by the condition
  \begin{equation}
    \label{eq:almunu}
    \d\alpha\wedge\d\mu \wedge\d\nu =0.
  \end{equation}
\end{prop}
\begin{proof}
  If $V$ is a Killing vector then $\cL_V \nu = 0$.  In a $G_2$
  solution there are two such independent vector field, which means
  that $\alpha,\mu,\nu$ must be functionally dependent.  Let us prove
  the converse.  We will show that the invariant count is $(0,1,2,2)$,
  which signifies a $G_2$ solution by the Karlhede algorithm.  By
  Proposition \ref{prop:alnonzero}, the second-order invariants are
  generated by $\mu, A, \nu$ and their complex conjugates. Suppose
  that
  \[ \d\alpha\wedge\d\mu\neq 0,\quad \d\alpha\wedge\d\baral = 0,\quad
  \d\alpha\wedge\d\mu \wedge\d\nu=0.\] By Propositions
  \ref{prop:albaldep} and \ref{prop:albaldep1},
  \[ B=0,\quad AA^*=1,\quad \d A = 0,\quad \barmu = A\mu.\] Hence, all
  second order invariants depend on $\alpha,\mu$.  The third order
  invariants are generated by $\delta^*A,\delta\mu,\Delta \mu,\Delta
  \nu$, and their complex conjugates.  Since $A$ is a constant and
  $\nu$ is a function of $\alpha, \mu$, and since relation
  \eqref{eq:deltanu} holds, it suffices to show that $\delta\mu$
  depend on $\alpha,\mu$. By Proposition \ref{prop:albaldep1} and by
  \eqref{eq:deltaalpha},
  \begin{equation}
    \label{eq:deltamu}
    \delta(\alpha\mu) = 0,\quad \delta \mu  + \baral \mu = 0,
  \end{equation}
  as was to be shown.
\end{proof}

\begin{prop}
  \label{prop:01nuDnu}
  Suppose that $\d\alpha\wedge\d\baral = \d\alpha\wedge\d\mu= 0$, but
  that $\d\alpha\wedge\d\nu\neq 0$. Then a $G_2$ solution is
  characterized by the conditions
  \begin{align}
    \label{eq:alnubnu}
    &\d\alpha\wedge\d\nu \wedge\d\nu^* =0,\\
    \label{eq:alnuDnu}
    &\d\alpha\wedge\d\nu \wedge\d\Delta\nu =0.
  \end{align}
\end{prop}

\begin{lem}
  \label{lem:b0aneq1}
  Suppose that $B=0, \mu \neq 0$.  The following are equivalent: (i)
  $\d\alpha \wedge\d\mu 
  = 0$ and(ii) $A=1, \Delta\mu=\mu^2$.
\end{lem}
\begin{proof}
  By assumption, $\barmu=A\mu$.  By Proposition \ref{prop:albaldep1},
  relation \eqref{eq:deltamu} holds. Hence,
  \begin{align*}
    \d\mu &= -\mu\baral \omega^1 -\alpha\mu\omega^2 + \Delta \mu\omega^3,\\
     d\alpha &= \alpha(\baral \omega^1 + A\alpha \omega^2 -\mu \omega^3),\\
     d\alpha\wedge \d\mu &= (A-1)\alpha^2\baral\mu\omega^1 \wedge
     \omega^2-\alpha\baral(A\mu^2 - \Delta\mu)\omega^1\wedge\omega^3
     \\
     &\quad -
     A \alpha^2(\mu^2-\Delta \mu)\omega^2\wedge\omega^3.
  \end{align*}
\end{proof}

\begin{proof}[Proof of Proposition \ref{prop:01nuDnu}]
  By Propositions \ref{prop:albaldep} \ref{prop:albaldep1}, $A$ is a
  constant.  Hence, using the reasoning in the proof of Proposition
  \ref{prop:01nu} above, $\nu^*,\Delta\mu, \Delta\nu$ generate the
  second and third-order invariants.  If $\mu\neq 0$, then by Lemma
  \ref{lem:b0aneq1}, $\Delta\mu$ is a function of $\mu$, which itself
  is a function of $\alpha$.  If $\mu = 0$, then afortiori $\Delta
  \mu=0$.  That means that $\nu^*, \Delta\nu$ generate all second and
  third-order invariants. Therefore, \eqref{eq:alnubnu}
  \eqref{eq:alnuDnu} suffice for a $G_2$ solution.
\end{proof}

We now classify the $(0,2,2)$ solutions.  Throughout, $V$ denotes the
2nd Killing vector independent from $D$.  The $G_2$ solutions can be
further subdivided according to whether $V^3\neq0$ or $V^3=0$.

By Proposition \ref{prop:Lprecursor}, the (0,2) precursor with $V^3=0$
is of class $\rL_{22}$.  The remaining (0,2) precursors are
$\rB_{23}$, $\rC_{23}$, $\rA_{23}$. As we show below, the
specialization from the precursor class to the $G_2$ class is governed
by the vanishing of the $Y$ and $\Upsilon$ invariants, which are
defined in \eqref{eq:Ydef} and \eqref{eq:Upsilondef}, respectively.
\begin{prop}
  \label{prop:B22}
  Suppose that $f(\zeta,u) = F(u^{-\i k} \zeta)u^{-2} + g u^{-2-\i
    k}\zeta,\; k\neq 0$ belongs to the $\rB_{23}$ precursor class.  The
  following are equivalent: (i) $\d\alpha\wedge\d\baral\wedge\d\nu=0$,
  (ii) $\hUpsilon=0$, (iii) $g'(u)=0$.
\end{prop}
\begin{proof}
  By Proposition \ref{prop:precursor}, $V= X_1^{-1} \hDelta$
  annihilates $\omega^1,\omega^2,\omega^3, \alpha,\mu$.  Above, we
  already noted that $\cL_V A = 0$.  By \eqref{eq:Bdef}
  \eqref{eq:Xdef}, $\cL_V X = 0$ also.  Let $\hnu$ be the invariant
  defined in \eqref{eq:hnudef}. By Proposition \ref{prop:om13} and
  \eqref{eq:bardelnu} \eqref{eq:deltanu} ,
  \begin{gather*}
    \delta X = -\baral X,\quad
    \bardel X = -\alpha X,\quad
    \Delta X = 2 X X_1,\\
    \delta(\hnu/X^*) = -4\i X_2,\quad
    \bardel(\hnu/X^*) = (1-2\hnu\alpha)/X^*,\\
    \hDelta(\hnu/X^*) = \Delta(\hnu/X^*) -4\i X X_2/\bardel +
    (1-2\hnu\alpha)/\alpha  = \hUpsilon
  \end{gather*}
  where $\hnu$ is the invariant defined by \eqref{eq:hnudef}.
  This proves the equivalence of (i) and (ii).   A direct calculation
  shows that
  \begin{gather*}
    \hUpsilon^* = 4u X_1^2/X\,  g'(u) F''(u^{-\i k} \zeta)^{-1/2}
  \end{gather*}
  This proves the equivalence of (ii) and (iii).
\end{proof}
\noindent
Remark 1: If $g'(u)=0$, then by \eqref{eq:zetaxform} we can absorb the
$g(u) u^{-2-\i k}\zeta$ term into the $F(u^{-\i
  k}\zeta)u^{-2}$ term.\\
Remark 2: the invariant $\hnu$ can be calculated directly by employing
the tetrad that respects the normalization $\hDelta \alpha =0$.  The
null rotation that sends $\Delta \to \hDelta$ maps $\nu \to \hnu$.

\begin{prop}
  \label{prop:C22} Suppose that $f(\zeta,u) = F(\e^{\i u} \zeta) + g
  \e^{\i u}\zeta,$ belongs to the $\rC_{23}$ precursor class.  The
  following are equivalent: (i) $\d\alpha\wedge\d\baral\wedge\d\nu=0$,
  (ii) $\hUpsilon=0$, (iii) $g'(u)=0$.
\end{prop}
\begin{proof}
  The proof is similar to the argument employed in Proposition
  \ref{prop:B22} above.  The formulas that differ are
  \begin{gather*}
    \Delta X = 0,\qquad
    \Upsilon^* = 4X^*\,  g'(u) F''(\e^{\i u} \zeta)^{-1/2}
  \end{gather*}
\end{proof}

\begin{prop}
  \label{prop:L22} Suppose that $f(\zeta,u) = g_1\log\zeta+g_2\zeta$
  belongs to the logarithmic $\rL_{23}$ precursor class.  The
  following are equivalent: (i) $\d\alpha\wedge\d\baral \wedge\d\nu =
  0$, (ii) $Y=0$, (iii) $g_2=0$.
\end{prop}
\begin{proof}
  By Proposition \ref{prop:Lprecursor}, $V=\im
  (\alpha^{-1}\bardel)$ annihilates $\omega^1,\omega^3,\alpha,\mu$.
  Hence, condition (i) is equivalent to $\cL_V\nu=0$. We have
  \begin{equation}
    \label{eq:Ydeltanu}
    (\delta\nu)/\baral -  (\delta^*\nu)/\alpha = -1/\alpha + 2\nu +
    (\mu^2+\Delta\mu) /\baral = Y
  \end{equation}
  This proves the equivalence of (i) and (ii). 
  A direct calculation shows that
  \[ \alpha Y^* = -\zeta g_2 (g_1 g_1^*)^{-1/2}.\]
  This proves the equivalence of (ii) and (iii).
\end{proof}
\begin{prop}
  \label{prop:A22} 
  Suppose that $f(\zeta,u) =F(\zeta)+g\zeta$ belongs to the $\rA_{23}$
  precursor class. The following are equivalent: (i)
  $\d\alpha\wedge\d\baral \wedge\d\nu  = 0$, (ii)  $\Delta\nu=0$,
  (iii)  $g'(u)=0$.
\end{prop}
\begin{proof}
  By Proposition \ref{prop:Aprecursor} a multiple of $\Delta$
  annihilates $\omega^1,\omega^3,\alpha,\baral$.  Hence (i) is equivalent to
  (ii).  A direct calculation shows that
  \[ \Delta \nu^* = \e^{-2a} g'(u)/F''(\zeta).\]
  This proves the equivalence of (ii) and (iii).
\end{proof}
\noindent Note that if $g'(u)=0$, then we can
absorb the $g\zeta$ term into the $F(\zeta)$ term.

We now classify the $G_2$ solutions of type $(0,1,2,2)$.  By
definition, these are specializations of the type $(0,1)$ precursors.
The latter solutions fall into three groups: (i) $V^3=0$, (ii)
$V^3\neq 0$ and $\d\alpha\wedge\d\mu=0$, (iii) $V^3\neq 0$ and
$\d\alpha\wedge\d \mu \neq 0$, where $V$ is the vector field that
satisfies \eqref{eq:precursor}. Case (i) is class $\rL_{23}$.  The
specialization to a $G_2$ solution is described, mutatis mutandi, by
Proposition \ref{prop:L22} above.  Case (ii) consists of classes
$\rL_{13}$, $\AP_{13}$, $\AE_{13}$, $\AE_{13}$. The specialization to
$G_2$ solutions is described by Propositions in \ref{prop:BL122},
\ref{prop:AP122}, \ref{prop:AE122}, \ref{prop:AL122} of the following
section. Case (iii) consists of classes $\BP_{13}$, $\CP_{13}$,
$\BE_{13}$, $\CE_{13}$.  By Proposition \ref{prop:01nu}, the
specialization to a $G_2$ solution is characterized by the condition
$\d\alpha\wedge\d\mu\wedge\d\nu=0$. The following Proposition analyzes
this condition.  The key invariant here is $\tUpsilon$, as defined by
\eqref{eq:tUpsilondef}.

\begin{lem}
  \label{lem:albaldep2} Suppose that $B=0$ and $AA^*=1, A\neq 1.$ Then
  \begin{equation}
    \label{eq:cLValmu}
    \{ \d\alpha,\d\alpha^*, \d\mu,\d\mu^* \}^\perp=
    \lspan \{ \tDelta, D \},  
  \end{equation}
  with $\tDelta$ defined as in \eqref{eq:tDeltadef}.  
\end{lem}
\begin{proof}
  Since $M = \alpha \mu$, no generality is lost if replace $\d\mu$ with
  $\d M$.  By Proposition \ref{prop:albaldep1}, $\delta M = 0$.  By
  \eqref{eq:bardeltamu}
  \[ \bardel M = (A-1) \alpha M.\] By \eqref{eq:Bdef}, $\barmu =
  A\mu$.   Hence, by \eqref{eq:deltaalpha} \eqref{eq:Deltaalpha} we seek
  the kernel of the following matrix:
  \begin{equation}
    \label{eq:tDeltasystem}
    \begin{pmatrix}
      \alpha\baral &  \alpha^2 A &  -AM &0\\
      \baral^2 A^{-1} & \alpha\baral & -M \baral \alpha^{-1} & 0\\
       0 & (A-1) \alpha M & \Delta M &0
    \end{pmatrix}
  \end{equation}
  By Proposition \ref{prop:albaldep} $\d\alpha\wedge\d\baral = 0$;
  hence, the above matrix has rank $2$.  Since $A^* = 1/A$, the kernel
  is invariant under complex conjugation.  Therefore, since $A\neq 1$,
  a basis for the kernel is $D$ and
  \[ \tDelta = \tX/\baral\, \delta + \tX^*/\alpha\, \bardel + \Delta +
  \tX\tX^*/(\alpha\baral) D,\quad \tX^* = \Delta M /(M(1-A)) \]
\end{proof}
\begin{prop}
  \label{prop:BCPE13}
  Suppose that $f(\zeta,u)=(k_0 z)^{\i k_1}u^{-2} + gu^{-2} z$, or
  $f(\zeta,u) = \exp(z) + g z$ where $z=u^{-\i k} \zeta$ or $z=\e^{\i
    u} \zeta$; i.e., $f(\zeta,u)$ belongs to one of the following
  classes: $\BP_{13}$, $\CP_{13}$, $\BE_{13}$, $\CE_{13}$.  Then, the
  following are equivalent: (i) $\d\alpha\wedge\d\mu\wedge\d\nu = 0$,
  (ii) $\tUpsilon=0$, (iii) $g'(u)=0$.
\end{prop}
\begin{proof}
  By assumption, $B=0, AA^*=1, A\neq 1$.  Hence,  there exists a $V$
  such that condition \eqref{eq:precursor}
  holds.  Since $\cL_V\alpha= \cL_V\mu=0$, by Lemma
  \ref{lem:albaldep2} $V$ is a multiple of $\tDelta$.  
  % Hence By assumption, $V$ annihilates $\alpha,\mu$ and hence also
  % $\cL_V M=0$.
  % Observe that
  % \begin{gather*}
  %   \mu^*  = A\mu,\quad A^* = 1/A\\
  %   \delta M = 0,\quad \bardel M = \alpha M (A-1) \\
  %   \tX = \Delta M/(M(1-A)) = (\Delta \mu - A \mu^2)/(\mu(1-A))
%     \intertext{Hence,}
%     \tDelta M =  \Delta M+ (\tX^*/\alpha)\,
%     \bardel M = 0\\
%     \tDelta \alpha = -A \alpha \mu+ A \alpha\frac{\Delta
%     \mu-A\mu^2}{\mu(1-A)} + \alpha\frac{A\Delta\mu - A\mu^2}{\mu(A-1)}= 0
%   \end{gather*}
%   by construction.  Since $\alpha,\mu$ are independent, it follows
%   that $V$ is a scalar multiple of $\tDelta$.  
  Hence, $\tX/\tX^* = C/C^*$ where $C=\baral V^1$, and hence $V =
  C/\tX \tDelta$.  In the proof of Proposition \ref{prop:Delta1mu} we showed
  that $\Delta(1/\mu)$ is a constant. It follows that $\cL_V
  \Delta\mu=0$ and hence $\cL_V \tX=0$ also.  Therefore, the desired
  condition is equivalent to $\tDelta (\tnu/\tX) = 0$ where $\tnu$ is
  the invariant defined in \eqref{eq:tnudef}.  By \eqref{eq:bardelnu},
  \eqref{eq:deltanu} \eqref{eq:deltaV3}, \eqref{eq:DeltaV3}
  \begin{gather*}
    \delta \tX = -\baral \tX,\quad
    \bardel \tX = -\alpha \tX,\quad
    \Delta \tX = 2 \tX \tX_1,\\
    \delta(\tnu/\tX) = -4\i \tX_2,\quad
    \bardel(\hnu/\tX) = (1-2\hnu\alpha)/\tX,\\
    \hDelta(\hnu/\tX) = \Delta(\hnu/\tX) -4\i \tX \tX_2/\bardel +
    (1-2\hnu\alpha)/\alpha  = \tUpsilon
  \end{gather*}
  This proves the equivalence of (i) and (ii).  A direct calculation
  shows that
  \[ \tUpsilon = C\alpha\mu^2 u^{1+\i k_1} \zeta^* (g_1'(u))^*,\]
  where $C=C(k_0,k_1)$ is a constant.  This proves the equivalence of
  (ii) and (iii).
\end{proof}
\noindent
Remark 1: If $g'(u)=0$, then by \eqref{eq:zetaxform} we can absorb the
the 2nd term in $f(\zeta,u)$ into the first term.  Remark 2: the
invariant $\tnu$ can be calculated directly by employing a
null-rotated tetrad that sends $\Delta \to \tDelta$ and $\nu \to
\tnu$.

\section{The maximal IC order  class. } 
\label{sect:0123}
This section is devoted to the proof of Theorem \ref{thm:maxorder};
we exhibit and classify all vacuum pp-wave solutions with a
$(0,1,2,3)$ invariant count.  The $(0,1)$ class is defined by the
condition $\d\alpha\wedge \d\baral=0$.  If $\alpha,\mu$ are
independent, then the $(0,1,2)$ condition requires that $\nu, \barnu$
be functions of $\alpha,\mu$.  However, by Proposition
\ref{prop:01nu}, this forces a $G_2$ solution, and therefore can be
excluded from the $(0,1,2,3)$ classification.

\noindent Thus, we have narrowed the search for  $(0,1,2,3)$
solutions  to the following class:
\begin{equation}
  \label{eq:0123class}
  \d\alpha\wedge\d\baral = 0, \quad
  \d\alpha\wedge\d\mu = 0, \quad \d\alpha\wedge\d\nu \wedge \d\barnu =0
\end{equation}
The middle condition forces some restrictions.
\noindent 
By Lemma \ref{lem:b0aneq1}, the analysis divides into two cases: $B=0,
A=1,\Delta\mu=\mu^2, \mu \neq0$ and $\mu=0, AA^*=1$.  The former
possibility specifies class $\BL_{13}$; the latter classes $\AP_{13},
\AE_{13}, \AL_{13}$. We begin by describing the specialization from
class $\BL_{13}$ to class $\BL_{123}$. The $Y$ invariant employed
below is defined in \eqref{eq:Ydef}.
\begin{prop}
  \label{prop:BL123}
  Suppose that $f(\zeta,u) = C u^{-2} \log \zeta +g\zeta$ belongs to
  class $\BL_{13}$.  The following are equivalent: (i)
  $\d\alpha\wedge\d \nu\wedge\d\nu^*=0$, (ii) $\Delta\log(Y Y^*)=4\mu$,
  (iii) $g = ku^{-2}\e^{\i h}$, where $k$ is a real constant and
  $h=h(u)$ is real.
\end{prop}
\begin{proof} 
  Our assumption implies
  \begin{gather*}
    A=1, \quad B=0\\
    \delta\mu = -\mu\baral,\quad
    \Delta\mu = \mu^2,\\
    Y = 2\nu-1/\alpha+ 2\mu^2/\baral.
  \end{gather*}
  Hence,  by \eqref{eq:bardelnu}  \eqref{eq:deltanu}
  \begin{gather*}
    \delta Y = -Y \baral,\quad
    \bardel Y  = -3Y \alpha,\\
    \begin{vmatrix}
      \delta\alpha & \bardel\alpha & \Delta \alpha\\
      \delta Y & \bardel Y & \Delta Y\\
      \delta Y^* & \bardel Y^* & \Delta Y^*\\
    \end{vmatrix} = 
    \begin{vmatrix}
      \alpha\baral & \alpha^2 & -\alpha\mu\\
      -Y\baral & -3 Y \alpha & \Delta Y\\
      -3 Y^* \baral & -\alpha  Y^* & \Delta Y^*\\
    \end{vmatrix} = 2\alpha^2\baral( 4 Y Y^* \mu - \Delta( Y Y^*))
  \end{gather*}
  This proves the equivalence of (i) and (ii).  Writing $g=\e^{h_1+\i
    h_2}$, a direct calculation shows that
  \begin{gather}
    \mu = -(CC^*)^{1/4}(\zeta\zeta^*)^{1/2},\\
    M=\alpha\mu =  (\i/2) (C^*)^{-1/2},\\
    YY^*=4 \e^{2h_1} u^4 \mu^4,\\
    \label{eq:DeltalogY}
    (\Delta\log Y Y^*)\mu = -2uh_1'(u).
  \end{gather}
  Therefore, (ii) is equivalent to 
  \[ uh_1'(u) = -2,\]
  which is equivalent to (iii).
\end{proof}
\noindent
We now prove that generically the above solution is (0,1,2,3), and
in the process derive the condition for specialization to a $G_2$ solution.
\begin{prop}
  \label{prop:BL122}
  Suppose that $f(\zeta,u) = u^{-2}(C \log \zeta + k \e^{\i h}\zeta)$
  belongs to class $\BL_{123}$.  The following are equivalent: (i)
  $\d\alpha\wedge\d\nu \wedge \d\Delta \nu = 0$, (ii) $\Delta
  (\alpha\Delta \log Y) = 0$, (iii) $\e^{\i h} = u^{\i k_1}$, where
  $k_1$ is a real constant.
\end{prop}
\begin{proof}
  All of the relations given in the proof of Proposition
  \ref{prop:BL123} hold.  Furthermore, by \eqref{eq:alnzcom1} -
  \eqref{eq:alnzcom4} 
  \begin{gather*}
    \delta\Delta Y = -2\baral\Delta Y,\quad \bardel\Delta Y = -4\alpha
    \Delta Y.
  \end{gather*}
  Thus, a direct calculation shows that
  \[ \d\alpha \wedge \d Y \wedge \d\Delta Y =
  2\alpha^2\baral(Y\mu\Delta Y + \Delta Y^2 - Y\Delta^2 Y)
  \omega^1\wedge\omega^2\wedge\omega^3.\]
  Since $\alpha\mu$ is a constant, the factor on the right can be
  written as
  \[ (Y\mu\Delta Y + \Delta Y^2 - Y\Delta^2 Y) =  Y^2 \alpha^{-1}
  \Delta(\alpha \Delta \log Y).\]
  This proves the equivalence of (i) and (ii).  Furthermore, a direct
  calculation gives
  \[ 2 C^{1/4} (C^*)^{3/4} \Delta(\alpha \Delta \log Y)= u (\zeta\zeta^*)^{1/2}(h_2'(u) +
  uh_2''(u)).\]
  This proves the equivalence of (ii) and (iii).
\end{proof}

We now consider the case of $\mu=0, AA^*=1$.
\begin{prop}
  \label{prop:AP123}
  Suppose that $f(\zeta,u) = (k_0\zeta)^{2\i k_1} + g\zeta$ belongs to
  the $\AP_{13}$ class.  The following are equivalent: (i)
  $\d\alpha\wedge\d\nu \wedge \d \nu^* = 0$, (ii)
  \begin{equation}
    \label{eq:3Arelation}
    (1-3A)\Delta\log Y + (A-3) \Delta \log Y^* = 0, 
  \end{equation}
 (iii) $g= k_2 \e^{\i h(1-2\i k_1)}$ where $k_2$ is
  a real constant and $h=h(u)$ is real.
\end{prop}
\begin{proof} 
  Our assumption and Proposition \ref{prop:albaldep1} imply that
  $\mu=0$ and that $A$ is a constant satisfying $A A^*=1$.  Hence, by
  \eqref{eq:bardelnu} and \eqref{eq:deltanu}
  \begin{gather*} 
    Y = (3-A)\nu - 1/\alpha\\
    \delta Y = -Y \baral,\quad \bardel Y = -3Y \alpha,\\
    \begin{vmatrix}
      \delta\alpha & \bardel\alpha & \Delta \alpha\\
      \delta Y & \bardel Y & \Delta Y\\
      \delta Y^* & \bardel Y^* & \Delta Y^*\\
    \end{vmatrix} = 
    \begin{vmatrix}
      \alpha\baral & A\alpha^2 & 0\\
      -Y\baral & -3 Y \alpha & \Delta Y\\
      -3 Y^* \baral & -\alpha  Y^* & \Delta Y^*\\
    \end{vmatrix} \\
    \qquad = 2\alpha^2\baral( (1-3A)Y^*\Delta Y + (A-3) Y \Delta Y^*)
  \end{gather*}
  This proves the equivalence of (i) and (ii). Writing
  \[ g = \e^{(1-2\i k_1)(h_1+ \i h_2)}, \]
  a direct calculation
  shows that
  \begin{gather*}
    \baral((1-3A)\Delta \log Y + (A-3)  \Delta \log Y^*) =C
    \zeta^{-\i k_1}h_1'(u),
  \end{gather*}
  where $C=C(k_0,k_1)$ is a constant.   This proves the equivalence of
  (ii) and (iii).
\end{proof}
\noindent
We now prove that generically the above solution is (0,1,2,3), and
in the process derive the condition for specialization to a $G_2$ solution.
\begin{prop}
  \label{prop:AP122}
  Suppose that $f(\zeta,u) = (k_0\zeta)^{2\i k_1} + k_2 \e^{\i h(1+2\i
    k_1)}\zeta$ belongs to class $\AP_{123}$.  The following are
  equivalent: (i) $\d\alpha\wedge\d\nu \wedge \d\Delta \nu = 0$, (ii)
  $\Delta^2 Y^{\frac{1-A}{A-3}} = 0$, (iii) $f(\zeta,u) = (k_0
  \zeta)^{2\i k_1} + C u^{-2-\i k_1} \zeta$, where $C$ is a complex
  constant.
\end{prop}
\begin{proof}
  All of the relations given in the proof of Proposition
  \ref{prop:AP123} hold.  Furthermore, by \eqref{eq:alnzcom1} -
  \eqref{eq:alnzcom4} ,
  \begin{gather*}
    \delta\Delta Y = -2\baral\Delta Y,\quad \bardel\Delta Y = -4\alpha
    \Delta Y.
  \end{gather*}
  From there, a direct calculation shows that
  \begin{gather*}
    \begin{vmatrix}
      \delta\alpha & \bardel\alpha & \Delta \alpha\\
      \delta Y & \bardel Y & \Delta Y\\
      \delta \Delta Y & \bardel \Delta Y & \Delta^2 Y\\
    \end{vmatrix} = 
    \begin{vmatrix}
      \alpha\baral & A\alpha^2 & 0\\
      -Y\baral & -3 Y \alpha & \Delta Y\\
      -2\baral \Delta Y & -4\alpha  \Delta Y & \Delta^2 Y\\
    \end{vmatrix} =\\
    \qquad = 2\alpha^2\baral( (2(2-A)(\Delta Y)^2 + (A-3) \Delta^2
    Y)=2\alpha\baral Y^{\frac{3A-7}{A-3}} \frac{(A-3)^2}{1-A} \Delta^2
    Y^{\frac{1-A}{A-3}}
  \end{gather*}
  This proves the equivalence of (i) and (ii).  Furthermore, a direct
  calculation gives
  \[ Y^{\frac{A-1}{A-3}}\Delta^2 Y^{\frac{1-A}{A-3}} = \tC
  \zeta^{1-\i k_1} \zeta^{1+\i k_1} (k_1 h_2'(u)^2-h_2''(u)),\]
  where $\tC$ is a complex constant.
  This proves the equivalence of (ii) and (iii).
\end{proof}
\noindent Finally, we consider the $\AE$ and the $\AL$ classes.
Propositions \ref{prop:AE123} and \ref{prop:AL123} derive the form of
the $(0,1,2)$ solutions for the cases $A=-1$ and $A=1$, respectively.
Propositions \ref{prop:AE122} and \ref{prop:AL122} prove that these
solutions are generically of type $(0,1,2,3)$ and derive the condition
for the specialization to the corresponding $(0,1,2,2)$ $G_2$
solution. Mutatis mutandi, these Propositions are proved in the same
way as Propositions \ref{prop:AP123} and \ref{prop:AP122} above.
\begin{prop}
  \label{prop:AE123}
  Suppose that $f(\zeta,u) = \exp(k\zeta) + g\zeta$ belongs to the
  $\AE_{13}$ class.  The following are equivalent: (i)
  $\d\alpha\wedge\d\nu \wedge \d \nu^* = 0$, (ii)
  $\Delta\log(Y/Y^*)=0$, (iii) $g= \e^{\i k_1}\e^h$ where $k_1$ is a
  real constant and $h=h(u)$ is real.
\end{prop}
\begin{prop}
  \label{prop:AE122}
  Suppose that $f(\zeta,u) = \exp(k_0\zeta) + \e^{\i k_1}\e^h\zeta$
  belongs to the $\AE_{123}$ class.  The following are equivalent: (i)
  $\d\alpha\wedge\d\nu \wedge \d \Delta\nu = 0$, (ii) $\Delta^2
  Y^{-1/2}=0$, (iii) $\e^h=Cu^{-2}$ where $C$ is a complex constant.
\end{prop}
\begin{prop}
  \label{prop:AL123}
  Suppose that $f(\zeta,u) = \e^{\i k}\log\zeta + g\zeta$ belongs to
  the $\AL_{13}$ class.  The following are equivalent: (i)
  $\d\alpha\wedge\d\nu \wedge \d \nu^* = 0$, (ii)
  $\Delta\log(YY^*)=0$, (iii) $g= k_2\e^{\i h}$ where $k_2$ is a real
  constant and $h=h(u)$ is real.
\end{prop}
\begin{prop}
  \label{prop:AL122}
  Suppose that $f(\zeta,u) = \e^{\i k_0}\log\zeta + k_1 \e^{\i h}\zeta$ belongs to
  the $\AL_{123}$ class.  The following are equivalent: (i)
  $\d\alpha\wedge\d\nu \wedge \d \Delta\nu = 0$, (ii)
  $\Delta^2 \log Y=0$,
 (iii) $\e^{\i h}=u^{\i k_2}$ where $k_2$ is a real constant.
\end{prop}

Finally we remark that a suitable change of variable
\eqref{eq:zetaxform} \eqref{eq:fxform} allows for two equivalent
representation for solution classes $\AP_{122}, \AE_{122}, \AL_{122}$:
\begin{align}
  &(k_0 \zeta)^{2\i k_1}+ C u^{-2-\i k_1} \zeta \simeq
  \left(k_0 u^{-\i/k_1}\zeta +C\right)^{2\i k_1} u^{-2}\\
  &\exp(k_0 \zeta)+Cu^{-2}\zeta \simeq (\exp(k_0\zeta)+C\zeta)u^{-2} \\
  &\e^{\i k_0}\log\zeta+k_1 \e^{\i u} \zeta \simeq e^{\i k_0}
  \log(e^{\i u} \zeta+k_1)
\end{align}
It follows that classes $\AP_{122}, \AE_{122}$ are specializations of
the generic $G_2$ solution $\rB_{22}$, while $\AL_{122}$ is a
specialization of $\rC_{22}$.

\section{The $G_3$ solutions}
\label{sect:g3}
In this section we classify the $G_3$ solutions. The invariant count
is $(0,1,1)$ and hence these solutions are characterized by
$\alpha\neq 0$ and
\[ \d\alpha\wedge\d\baral = \d\alpha\wedge\d\mu = \d\alpha \wedge\d A
= \d\alpha \wedge\d\nu = 0.\] The condition $\d\alpha \wedge\d A=0$ is
redundant, because by Propositions \ref{prop:albaldep} and
\ref{prop:albaldep1}, a $G_3$ solution satisfies $B=0, A A^*=1, \d A =
0$.  By Lemma \ref{lem:b0aneq1} there are two branches: (i)
$B=0, A=1, \Delta \mu=\mu^2,\; \mu\neq 0$; and (ii) $\mu=0, A A^*=1$.
By Propositions \ref{prop:BL123}, \ref{prop:AP123}, \ref{prop:AE123},
\ref{prop:AL123} the condition $\d\alpha\wedge\d\nu\wedge\d\nu^*=0$,
which is weaker than $\d\alpha\wedge\d\nu=0$, specializes these two
branches to $(0,1,2,3)$ solutions. Therefore the $G_3$ solutions arise
as the following sequence of specializations:
\[ (0,1,3) \to (0,1,2,3)\to (0,1,2,2) \to (0,1,1).\] Therefore, to
classify the $G_3$ solutions it suffices to begin with the classes
$\BL_{13}$, $\AP_{13}$, $\AE_{13}$, $\AL_{13}$ and impose the
specialization is $\d\alpha\wedge\d\nu=0$.
\begin{prop}
  \label{prop:BL11}
  Suppose that $f(\zeta,u) = Cu^{-2}\log\zeta + gu^{-2}\zeta$
  belongs to class $\BL_{13}$. The following are equivalent: (i)
  $\d\alpha\wedge\d\nu = 0$, (ii) $Y=0$, (iii) $g=0$.
\end{prop}
\begin{proof}
  Using the relations from the proof of Proposition \ref{prop:BL123},
  we have
  \begin{gather*}
    \delta\alpha \bardel Y - \delta Y \bardel \alpha = -2Y \alpha^2
    \baral,\\
    \delta\alpha \Delta Y - \delta Y \Delta \alpha = -2
    \alpha\baral(Y\mu-\Delta Y),\\
    \Delta\alpha \bardel Y - \Delta Y \bardel \alpha =  \alpha^2
    (3Y\mu-\Delta Y)
    \end{gather*}
    This proves the equivalence of (i) and (ii).  A direct calculation
    shows that
    \[ \baral Y^* =  u^2 \zeta g/C.\]
    This proves the equivalence of (ii) and (iii).
\end{proof}
\begin{prop}
  \label{prop:AP11}
  Suppose that $f(\zeta,u) = (k_0 \zeta)^{2\i k_1} +g \zeta$ belongs
  to the $\AP_{13}$ class. The following are equivalent: (i)
  $\d\alpha\wedge\d\nu = 0$, (ii) $Y=0$, (iii) $g=0$.
\end{prop}
\begin{proof}
  Using the relations from the proof of Proposition \ref{prop:AP123},
  we have
  \begin{gather*}
    \delta\alpha \bardel Y - \delta Y \bardel \alpha = (A-3)Y \alpha^2
    \baral,\\
    \delta\alpha \Delta Y - \delta Y \Delta \alpha = \alpha\baral
    \Delta Y,\\
    \Delta\alpha \bardel Y - \Delta Y \bardel \alpha = A\alpha^2
    \Delta Y
    \end{gather*}
    This proves the equivalence of (i) and (ii).  A direct calculation
    shows that
    \[ \alpha^* Y^* =  C\zeta^{1-2\i k_1} g,\]
    where $C$ is a constant.
    This proves the equivalence of (ii) and (iii).
\end{proof}
The proof of the following two Propositions uses the same argument as
above.  One merely specializes $A\to -1$ and $A\to 1$, respectively.
\begin{prop}
  \label{prop:AE11}
  Suppose that $f(\zeta,u) = \exp(k\zeta) +g \zeta$ belongs
  to the $\AE_{13}$ class. The following are equivalent: (i)
  $\d\alpha\wedge\d\nu = 0$, (ii) $Y=0$, (iii) $g=0$.
\end{prop}
\begin{prop}
  \label{prop:AL11}
  Suppose that $f(\zeta,u) = \e^{\i k} \log \zeta +g \zeta$ belongs
  to the $\AL_{13}$ class. The following are equivalent: (i)
  $\d\alpha\wedge\d\nu = 0$, (ii) $Y=0$, (iii) $g=0$.
\end{prop}

\section{The $G_5$ and $G_6$ solutions}
\label{sect:g56}

In this section we derive and classify the metric forms in the
$\alpha  = 0$ class.  By Proposition \ref{prop:alzero} the
corresponding solutions are either $G_5$ or $G_6$.

\begin{prop}
  The following are equivalent: (i) $\alpha=0$ and (ii) $f(\zeta,u) =
  g_2 \zeta^2+g_1 \zeta + g_0$, where as usual $g_i=g_i(u),\;
  i=0,1,2$ denote complex valued functions of one variable.
\end{prop}
\begin{proof}
  A direct calculation shows that
  \[ \alpha = e^{a-a^*} (a_\zeta)^*,\]
  where
  \[ a=\frac{1}{4} f_{\zeta\zeta}.\]
\end{proof}
\noindent
Note that a form-preserving transformation \eqref{eq:zetaxform} --
\eqref{eq:fxform} can be used to set $g_1, g_0\to 0$.  Hence, without
loss of generality a solution in the $\alpha=0$ class has the form
$f(\zeta,u) = g \zeta^2$, where $g\neq 0$.

It will be convenient to set $g=e^{4A}$, where $A=A(u)$ is complex
valued.  A direct calculation then shows that
\begin{align}
  \label{eq:al0gamma}
  \gamma&= \frac{e^{-2\Re A} A^*_u}{\sqrt{2}}\\
  \label{eq:gambgam}
  \frac{\gamma}{\gamma^*} &= \frac{A^*_u}{A_u}.
\end{align}
We are now in a position to derive and classify the homogeneous $G_6$
solutions. Such solutions are characterized by the condition
$\Delta\gamma=0$, which ensures that the fundamental Cartan invariant
$\gamma$ is a constant.

At this point the $G_6$ classification bifurcates, depending on the
value of $A_u$.  We consider the generic case in Proposition
\ref{prop:B0}, and the singular case in Proposition \ref{prop:C0}.
The classification is summarized in Table \ref{tab:g56}.
\begin{prop}
  \label{prop:B0}
  Suppose that $f(\zeta,u) = e^{4A} \zeta^2$, $\Delta\gamma=0$, and
  $\Re\gamma\neq 0$.  Then, without loss of generality,
  \begin{equation}
    f(\zeta,u) = k_1 u^{2\i k_0-2} \zeta^2.
  \end{equation}
\end{prop}
\begin{proof}
  If $\Delta\gamma=0$, then $\gamma$ is a constant.  By assumption,
  $A_u \neq 0$ , and so $\gamma/\gamma^*$ is also a constant.  It will
  therefore be convenient to write
  \begin{equation}
    \label{eq:al0hkdef}
    1/A_u = e^{\i k} h,
  \end{equation}
  where both $k$ is a real constant and $h=h(u)$ is real.
  A direct calculation now gives
  \begin{align*}
    h_u &= -\frac{1}{2}\cos k,\\ \intertext{which implies }
    A_u &= \frac{2\e^{-\i k}}{k_2 - u \cos k},\\
    f &= (\cos k u - k_2)^{-2+2\i \tan k} k_1
  \end{align*}
  where $k_1\neq 0$ is a real constant.  Substituting into
  \eqref{eq:al0gamma} gives
  \[ \gamma = \frac{e^{\i k}}{\sqrt{8 k_1}},\]
  which means that $k,k_1$ are essential constants, while $k_2$ can be
  gauged away.  Applying the change of variables \eqref{eq:uxform}
  gives the desired solution form.
\end{proof}

\begin{prop}
  \label{prop:C0}
  Suppose that $f(\zeta,u) = e^{4A} \zeta^2$, $\Delta\gamma=0$, and
  $\Re\gamma= 0$.  Then, without loss of generality,
  \begin{equation}
    f(\zeta,u) = \e^{2\i k_0 u} \zeta^2,
  \end{equation}
  where $k_0$ is a real constant.
\end{prop}
\begin{proof}
  The super-singular case of $\gamma=0$ corresponds to $A_u=k=0$.  From
  now on, we suppose that $\gamma$ is a non-zero imaginary constant.
  It follows that 
  \[ A_u = \i k \]
  where $k$ is some real constant.  The desired conclusion follows immediately.
\end{proof}
\section{Conclusions} \label{conclusions}

In our search for those vacuum PP-wave spacetimes in which the
fourth-order covariant derivatives of the curvature tensor are
required to classify them entirely, we have produced an approach to
invariantly classifying the vacuum PP-wave spacetimes.  Our approach
is based on Cartan invariants and the Karlhede algorithm and is
necessitated by the fact that a the class of vacuum PP-waves has
vanishing scalar invariants \cite{ACSI}. Our classification is finer
than the analysis of each spacetime's isometry group alone. The
summary of this invariant approach to classification is given in
tables \ref{tab:022} -- \ref{tab:011} with specialization relations
summarized in Figures \ref{fig:g23} and \ref{fig:g1}.

For any spacetime, the classification begins with the fact that the
components of the curvature tensor and its covariant derivatives
produce all of the invariants required. The Karlhede algorithm
provides an algorithmic approach to determining the lowest order, $q$,
of covariant differentiation needed to classify the space, canonical
forms for the components of the curvature tensor and the number of
functionally independent invariants, $(t_0,t_1,\ldots, t_q)$ arising
from the collection of all components of the curvature tensor and its
covariant derivatives up to order $q$.  

For vacuum pp-waves we have demonstrated that $q\leq 4$ and have
classified all solutions that attain an IC order of $4$.  Table
\ref{tab:0123} summarizes the maximal order solutions.  By
characterizing the $G_2$ and $G_3$ solutions in terms of invariant
conditions, the invariant approach also sheds light on the origin of
the additional Killing vectors.  Another remarkable finding is the
fact that the maximal order solutions of Table \ref{tab:0123} are
direct precursors of the $G_3$ solutions first discovered by Kundt and
Ehlers.  In terms of the metric form, the mechanism of specialization
is the disappearance of an additive term; e.g.,
\[  e^{\i k_0} \log \zeta + k_1 e^{\i h} \zeta \to e^{\i k} \log
\zeta.\]

Outside of the invariant classification of spacetimes, the study of
the invariant structure of the Riemann tensor and its covariant
derivatives reveal the interconnection between spacetimes with less
symmetry and their more symmetric counterparts and how these arise as
specialization of the classifying manifold. Furthermore by imposing
conditions on the Cartan invariants we produced definite examples of
spacetime with little or no symmetry. This is particularly relevant
for the PP-wave spacetimes as before our work little was known about
those spacetimes admitting $D = \partial_v$ as the sole Killing
vector.

The approach used to invariantly classify the PP-waves is not limited
to this class alone. One may repeat the process for the other half of
the plane-fronted waves, the Kundt waves \cite{MMC}. Together these
spacetimes constitute the entirety of all Petrov type N VSI
spacetimes: the class of spacetimes where all scalar curvature
invariants vanish. These spacetimes are a special case of the CSI
spacetimes , where all scalar curvature invariants are constant, and
so the Karlhede algorithm is the only approach to invariantly
classifying these spaces.

Future research direction involve the extension of the invariant
classification to all VSI space-times, and even the full class of
Kundt-degenerate spacetimes.  The question of the physical and
phenomenological interpretation of the classifying invariants is also
unresolved, although some steps in this direction are ongoing
\cite{mcnutt}.

\section{Acknowledgements}
The authors would like to thank Georgios Papadopoulos for useful
discussions. The research of RM and AC is supported, in part, by NSERC
discovery grants.

\appendix
\section{Tables of exact solutions}
\label{sect:table}
Tables \ref{tab:g1} and \ref{tab:g101} summarize the exact solutions
derived in Section \ref{sect:g1}. Tables \ref{tab:preg202} and
\ref{tab:preg201}  summarize the precursor
solutions derived  Section in \ref{sect:g2precursor}.  Tables
\ref{tab:022} and \ref{tab:0122} give the $G_2$ solutions.

\begin{figure}[ht]
  \centering
\includegraphics[width=10cm]{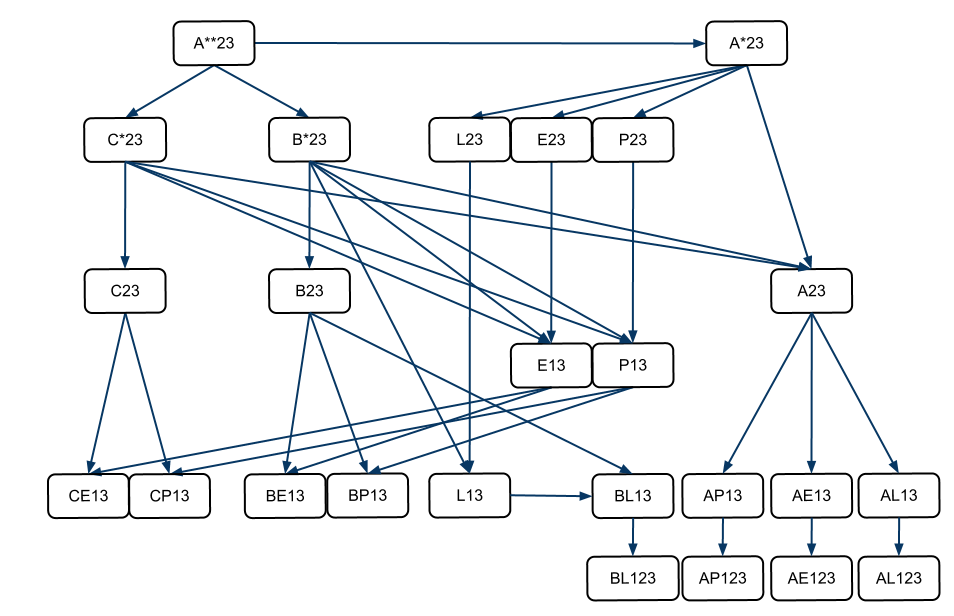}
  \caption{$G_1$ solutions}
  \label{fig:g1}
\end{figure}

\begin{table}[ht]
  \centering
  \begin{tabular}{|@{\vbox to 10pt{} } c|c|l|}
    \hline
    $G_1$ & $f(\zeta,u)$ & Invariant condition\\ \hline
    $\rA^{**}_{23}$ & $g_1 F(g_2 \zeta) + g_3\zeta$ &  $\delta \barA\,
    \delta^2 \!M = 
    \delta M\, \delta^2 \!\barA$\\
    $\rA^*_{23}$ & $F(g\zeta) g^{-2} + g_1 \zeta$ & $\delta M = 0$\\
    $\rA_{23}$& $F(\zeta)  + g_1 \zeta$ & $\mu= 0$\\
    $\rB^*_{23}$& $F(h^{\i k}\zeta)h^{2}+g\zeta$ &$B_2/B_1= k,\;
    A\barA\neq 1, B_1\neq 0$\\
    $\rC^*_{23}$ &$F(\e^{\i h}\zeta)+g\zeta$ & $B_1=0,\; A\neq 1$ \\
    $\rP_{23}$ &$(\e^{g_1}\zeta)^{\i h} + g_2 \zeta$ & $A\barA=1,\; A^2\neq 1$\\
    $\rE_{23}$ &$\exp(g_1\zeta) + g_2 \zeta$ & $A=-1$\\
    $\rL_{23}$ &$g_1\log\zeta + g_2 \zeta$ & $A=1$\\
    \hline
  \end{tabular}
  \caption{Type $(0,2,3)$ solution classes}
  \label{tab:g1}
\end{table}

\begin{table}[ht]
  \centering
  \begin{tabular}{|@{\vbox to 10pt{} } c|c|l|}
    \hline
    $G_1$ & $f(\zeta,u)$ & Invariant condition\\ \hline
    $\rP_{13}$ & $(k_0 \e^{\i h} \zeta)^{2\i k_1}+g\zeta$ &
    $B=0,\; A^2\neq 1,\; \mu \neq 0$,\\
%    $G_1$-ca $\AP_{13}$ & $(k_0 \zeta)^{2\i k_1}+g\zeta$ &
%    $\mu=0,\; A\barA=1,\; A^2\neq 1$,\\
    $\rE_{13}$ & $\exp(k \e^{\i h} \zeta)+g\zeta,$ &
    $B=0,\, A=-1$ \\ 
    $\rL_{13}$& $ \e^{\i k}h\log \zeta+g_2\zeta$ & $B=0, A=1$\\
    \hline
  \end{tabular}
  \caption{Type $(0,1,3)$ solutions}
  \label{tab:g101}
\end{table}

\begin{table}[ht]
  \centering
  \begin{tabular}{|@{\vbox to 10pt{} } c|c|l|}
    \hline
    $G_1$ & $f(\zeta,u)$ & Invariant condition\\ \hline
    $\rB_{23}$& $F( u^{-ik}\zeta)u^{-2}+g\zeta$ &$B_2/B_1=
    k,\;  \Delta X_1     = 2X_1^2$,  $AA^*\neq 1,\; B_1\neq 0$\\ 
    $\rC_{23}$ &$F(\zeta \e^{\i u})+g\zeta$ & $B_1=0,\Delta X_2=0,\;
    \mu\neq 0, AA^*\neq 1$ \\
    $\rL_{23}$ &$g_1\log\zeta + g_2 \zeta$ & $A=1$\\
    $\rA_{23}$ & $F(\zeta)  + g_1 \zeta$ & $\mu= 0$\\
    \hline
  \end{tabular}
  \caption{Type $(0,2,3)$  $G_2$-precursor solutions}
  \label{tab:preg202}
\end{table}

\begin{table}[ht]
  \centering
  \begin{tabular}{|@{\vbox to 10pt{} } c|c|l|}
    \hline
    $G_1$ & $f(\zeta,u)$ & Invariant condition\\ \hline
     $\BP_{13}$& $(k_0 u^{-\i k_1} \zeta)^{2\i k_2}+g\zeta$ &
    $B=0,\; \Delta^2(1/\mu) = 0,\;\Delta\mu\neq 0,\; A^2\neq 1$,\\
    $\CP_{13}$& $(k_0 \e^{\i u} \zeta)^{2\i k_1}+g\zeta$ &
    $B=0,\; \Delta\mu = 0,\; \mu\neq 0,\; A^2\neq 1$,\\
     $\BE_{13}$ & $\exp(k_0 u^{-\i k_1} \zeta)+g\zeta,$ &
    $B=0,\, A=-1, \; \Delta^2(1/\mu) = 0,\;\Delta\mu\neq 0$ \\ 
    $\CE_{13}$ & $\exp(k_0 \e^{\i u} \zeta)+g\zeta,$ &
    $B=0,\, A=-1, \Delta \mu = 0,\; \mu\neq 0$ \\ 
    $\rL_{13}$ & $ \e^{\i k}h\log \zeta+g_2\zeta$ & $B=0, A=1$\\
    $\BL_{13}$ & $ C u^{-2}\log \zeta+g\zeta$ & $B=0, A=1, \Delta
    \mu = \mu^2,\; \mu\neq 0$\\  
    $\AP_{13}$ & $(k_0 \zeta)^{2\i k_1}+g\zeta$ &
    $\mu= 0,\; A\barA =1,\; A^2\neq 1$,\\
    $\AE_{13}$ & $\exp(k \zeta)+g\zeta$ &
    $\mu= 0,\; A=-1$,\\
    $\AL_{13}$ & $\e^{\i k}\log \zeta+g\zeta$ &
    $\mu= 0,\; A=1$,\\
    \hline
  \end{tabular}
  \caption{Type $(0,1,3)$  $G_2$-precursor solutions}
  \label{tab:preg201}
\end{table}

\begin{table}[ht]
  \centering
  \begin{tabular}{|@{\vbox to 10pt{} } c|c|l|}
    \hline
    $G_1$ & $f(\zeta,u)$ & Invariant condition\\ \hline
    $\BL_{123}$ & $ (C\log \zeta+k\e^{\i h}\zeta)u^{-2}$ &
    $B=0, A=1, \Delta\mu=\mu^2, \Delta\log(Y Y^*)=4\mu,\; \mu\neq 0$\\
    $\AP_{123}$ & $(k_0 \zeta)^{2\i k_1}+k_2\e^{\i h(1-2\i k_1)}\zeta$ &
    $\mu= 0,\; A\barA =1,  \; A^2\neq 1$\\
    && $(1-3A)\Delta \log Y+(A-3)\Delta \log Y^*    = 0$\\
    $\AE_{123}$ & $\exp(k_0 \zeta)+ \e^{\i k_1}\e^{ h}\zeta$ &
    $\mu= 0,\, A=-1,\, \Delta (Y/Y^*)    = 0  $\\
    $\AL_{123}$ & $e^{\i k_0}\log\zeta+k_1 \e^{\i h}\zeta$ &
    $\mu= 0,\, A=1, \Delta(Y Y^*) = 0 $\\
    \hline
  \end{tabular}
  \caption{Type $(0,1,2,3)$  solutions}
  \label{tab:0123}
\end{table}

\begin{table}[ht]
  \centering
  \begin{tabular}{|@{\vbox to 10pt{} } c|c|l|}
    \hline
    $G_2$ & $f(\zeta,u)$ & Invariant condition\\ \hline
    $\BP_{122}$ & $((k_0 u^{-\i k_1} \zeta)^{2\i k_2}+k_3
    u^{-\i k_1} \zeta)u^{-2}$ &  
    $B=0,\; \Delta^2(1/\mu) = 0,\, \tUpsilon=0,\;\Delta\mu\neq 0,\;
    A^2\neq 1$ \\ 
    $\CP_{122}$ & $(k_0 \e^{\i u} \zeta)^{2\i k_1}+ k_2
    \e^{\i u} \zeta$ & 
    $B=0,\; \Delta\mu = 0,\, \tUpsilon=0,\; \mu\neq 0,\; A^2\neq 1$ \\
    $\BE_{122}$ & $\exp(k_0 u^{-\i k_1} \zeta)+k_2 u^{-\i
      k_1} \zeta,$ & 
    $B=0,\, A=-1, \; \Delta^2(1/\mu) = 0,\, \tUpsilon=0,\;\Delta\mu\neq 0$ \\ 
    $\CE_{122}$ & $\exp(k_0 \e^{\i u} \zeta)+k_1 e^{\i u} \zeta,$ &
    $B=0,\, A=-1, \Delta \mu = 0,\,\tUpsilon=0,\; \mu\neq 0$ \\ 
    $\rL_{122}$ & $ \e^{\i k}h\log \zeta$ & $B=0, A=1, Y=0$\\
    $\BL_{122}$ & $ u^{-2} (C\log \zeta+k\zeta)$ &
    $B=0, A=1, \Delta\mu=\mu^2,\;    \mu\neq 0$\\
    & &$\Delta\log(Y Y^*) = 4\mu,\;\Delta (\alpha \Delta\log Y)=0$\\ 
    $\AP_{122}$ & $(k_0 \zeta)^{2\i k_1}+ C u^{-2-\i k_1} \zeta$ &
    $\mu= 0,\; A\barA =1, \, \Delta^2 Y^{\frac{1-A}{A-3}} = 0,\;\;
    A^2\neq 1$,\\
    &$\left(k_0 u^{-\i/k_1}\zeta +C\right)^{2\i k_1} u^{-2}$&\\
    $\AE_{122}$ & $\exp(k_0 \zeta)+Cu^{-2}\zeta$ &
    $\mu= 0,\; A=-1,\, \Delta^2 Y^{-1/2} = 0$\\
    & $(\exp(k_0\zeta)+C\zeta)u^{-2}$ & \\
    $\AL_{122}$ & $\e^{\i k_0}\log\zeta+k_1 \e^{\i u} \zeta$ &
    $\mu= 0,\, A=1, \Delta^2\log Y = 0 $\\
    & $e^{\i k_0} \log(e^{\i u} \zeta+k_1)$ &\\
    \hline
  \end{tabular}
  \caption{Type $(0,1,2,2)$  $G_2$- solutions}
  \label{tab:0122}
\end{table}

\begin{table}[ht]
  \centering
  \begin{tabular}{|@{\vbox to 10pt{} } c|c|l|}
    \hline
    $G_3$ & $f(\zeta,u)$ & Invariant condition\\ \hline
    $\BL_{11}$ & $ C u^{-2} \log \zeta$ &
    $B=0, A=1, \Delta\mu=\mu^2, Y=0,\; \mu\neq 0$\\
    $\AP_{11}$ & $(k_0 \zeta)^{2\i k_1}$ &
    $\mu= 0,\; A\barA =1,\; Y = 0,\; A^2\neq 1$\\
    $\AE_{11}$ & $\exp(k_0 \zeta)$ &
    $\mu= 0,\, A=-1, Y = 0 $\\
    $\AL_{11}$ & $e^{\i k}\log\zeta$ &
    $\mu= 0,\, A=1, Y = 0 $\\
    \hline
  \end{tabular}
  \caption{Type $(0,1,1)$  $G_3$ solutions}
  \label{tab:011}
\end{table}

\begin{table}[ht]
  \centering
  \begin{tabular}{|@{\vbox to 10pt{} } c|c|l|}
    \hline
    Label & $f(\zeta,u)$ & Invariant condition \\ \hline
    $\rA_{11}$ & $g \zeta^2 $ & $\alpha = 0, \Delta\gamma\neq 0$\\
    $\rB_0$ & $k_1 u^{2\i k_0 -2} \zeta^2$ & $\alpha  = 0,\;\Delta\gamma =
    0,\; \Re\gamma \neq 0$  \\
    $\rC_0$ & $\exp(2\i k_0u) \zeta^2$ & $\alpha  = 0,\;\Delta\gamma =
    0,\; \Re\gamma = 0$  \\ \hline 
  \end{tabular}
  \caption{The $G_5$ and $G_6$ solutions}
  \label{tab:g56}
\end{table}

%\begin{section}{Rotating Plane-Fronted Waves}
%\input{RPFW.tex}
%\end{section}

\end{document}